\documentclass{article}
\usepackage{ijcai19}

\usepackage{times}
\usepackage{soul}
\usepackage{url}
\usepackage[hidelinks]{hyperref}
\usepackage[utf8]{inputenc}
\usepackage[small]{caption}
\usepackage{subcaption}
\usepackage{graphicx}
\usepackage{amsfonts}
\usepackage{amsmath}
\usepackage{amsthm}
\usepackage{xcolor}
\usepackage[makeroom]{cancel}
\usepackage[english]{babel}
\usepackage{booktabs}
\usepackage{algorithm}
\usepackage{algorithmic}
\usepackage{setspace}
\DeclareMathOperator*{\argmax}{arg\,max}  
\newtheorem{definition}{Definition}

\newtheorem{theorem}{Theorem}

\title{A Regularized Opponent Model with Maximum Entropy Objective}

\author{Zheng Tian$^1$\footnote{The first two authors contributed equally.}\and
Ying Wen$^{1*}$\and
Zhichen Gong$^{1}$\and
Faiz Punakkath$^{1}$\and
Shihao Zou$^{2}$ \And
Jun Wang $^{1}$\\
\affiliations
$^1$University College London\\
$^2$University of Alberta\\
\emails
\{zheng.tian, ying.wen, jun.wang\}@cs.ucl.ac.uk
}

\begin{document}

\maketitle

\begin{abstract}
In a single-agent setting, reinforcement learning (RL) tasks can be cast into an inference problem by introducing a binary random variable $o$, which stands for the ``optimality". In this paper, we redefine the binary random variable $o$ in multi-agent setting and formalize multi-agent reinforcement learning (MARL) as probabilistic inference. We derive a variational lower bound of the likelihood of achieving the optimality and name it as Regularized Opponent Model with Maximum Entropy Objective (ROMMEO). From ROMMEO, we present a novel perspective on opponent modeling and show how it can improve the performance of training agents theoretically and empirically in cooperative games. To optimize ROMMEO, we first introduce a tabular Q-iteration method ROMMEO-Q with proof of convergence. We extend the exact algorithm to complex environments by proposing an approximate version, ROMMEO-AC. We evaluate these two algorithms on the challenging iterated matrix game and differential game respectively and show that they can outperform strong MARL baselines.
\end{abstract}
\section{Introduction}
Casting decision making and optimal control as an inference problem have a long history, which dates back to~\cite{kalman1960} where the Kalman smoothing is used to solve optimal control in linear dynamics with quadratic cost. Bayesian methods can capture the uncertainties regarding the transition probabilities, the rewards functions in the environment or other agents' policies. This distributional information can be used to formulate a more structured exploration/exploitation strategy than those commonly used in classical RL, e.g. $\epsilon$-greedy. A common approach in many works~\cite{Toussaint:2006,Rawlik:2013,Levine2013,2018abbas}  for framing RL as an inference problem is by introducing a binary random variable $o$ which represents ``optimality". By this way, RL problems are able to lend itself to powerful inference tools~\cite{Levine2018review}. However, the Bayesian approach in a multi-agent environment is less well studied.

In many single-agent works, maximizing entropy is part of a training agent's objective for resolving ambiguities in inverse reinforcement learning~\cite{ziebart2008maximum}, improving the diversity~\cite{Florensa2017hierachy}, robustness~\cite{Fox2015taming} and the compositionality~\cite{harrnoja2018composable} of the learned policy. In Bayesian RL, it often presents in the evidence lower bound (ELBO) for the log likelihood of optimality~\cite{HaarnojaTAL17rldep,Schulman2017equivalence,haarnoja2018softactorcritic}, commonly known as maximum entropy objective (MEO), which encourages the optimal policy to maximize the expected return and long term entropy. 

In MARL, there is more than one agent interacting with a stationary environment. In contrast with the single agent environment, an agent's reward not only depends on the current environment state and the agent's action but also on the actions of others. The existence of other agents increases the uncertainty in the environment. Therefore, the capability of reasoning about other agents' belief, private information, behavior, strategy, and other characteristics is crucial. A reasoning model can be used in many different ways, but the most common case is where an agent utilize its reasoning model to help its self decision making~\cite{brown:fp1951,DBLP:journals/corr/HeinrichS16,He2016,Raileanu2018,wen2018probabilistic,tian2018birdge}. In this work, we use the word ``opponent” when referring to another agent in the environment irrespective of the environment’s cooperative or adversarial nature.

In our work, we reformulate the MARL problem into Bayesian inference and derive a multi-agent version of MEO, which we call the regularized opponent model with maximum entropy objective (ROMMEO). Optimizing this objective with respect to one agent's opponent model gives rise to a new perceptive on opponent modeling. We present
two off-policy RL algorithms for optimizing ROMMEO in MARL. ROMMEO-Q is applied in discrete action case with proof of convergence. For the complex and continuous action environment, we propose ROMMEO Actor-Critic (ROMMEO-AC), which approximates the former procedure and extend itself to continuous problems. We evaluate these two approaches on the matrix game and the differential game against strong baselines and show that our methods can outperform all the baselines in terms of the overall performance and speed of convergence.
\section{Method}
\subsection{Stochastic Games}
For an $n$-agent stochastic game~\cite{shapley1953stochastic}, we define a tuple $( \mathcal{S}, \mathcal{A}^1, \dots, \mathcal{A}^n, R^1, \dots, R^n, p, \mathcal{T}, \gamma )$, where $\mathcal{S}$ denotes the state space, 
$p$ is the distribution of the initial state, 
$\gamma$ is a discount factor, 
$\mathcal{A}^i$  and $R^i = R   ^i(s, a^i, a^{-i})$ are the action space and the reward function for agent $i \in \{1,\dots,n\}$ respectively. States are transitioned according to $\mathcal{T}:\mathcal{S}\times\mathcal{A}$, where $\mathcal{A}=\{\mathcal{A}^1,\cdots, \mathcal{A}^n\}$.
Agent $i$ chooses its action $a^{i} \in \mathcal{A}^{i}$ according to the policy $\pi^i_{\theta^i}(a^{i} | s)$  parameterized by $\theta^i$ conditioning on some given state $s \in \mathcal{S}$. Let us define the joint policy as the collection of all agents' policies $\pi_{\theta}$ with $\theta$ representing the joint parameter. It is convenient to interpret the joint policy from the perspective of agent $i$ such that $\pi_{\theta} = (\pi^{i}_{\theta^i}(a^{i}| s), \pi^{-i}_{\theta^{-i}}(a^{-i}| s))$, where $a^{-i} = (a^j)_{j \neq i}$, $\theta^{-i} = (\theta^j)_{j \neq i}$, and $\pi^{-i}_{\theta^{-i}}(a^{-i}| s)$ is a compact representation of the joint policy of all complementary agents of $i$. At each stage of the game, actions are taken simultaneously. Each agent is presumed to pursue the maximal cumulative reward,  expressed as
\begin{align}
\max~~\eta^i(\pi_{\theta}) = \mathbb { E } \left[ \sum _ { t = 1 } ^ { \infty } \gamma ^ { t } R^i( s_{t}, a^{i} _ { t }, a^{-i} _ { t  }  ) \right],
\label{marl_target}
\end{align}
with $(a^{i}_t, a^{-i}_t)$ sample from $(\pi^i_{\theta^i}, \pi^{-i}_{\theta^{-i}})$. In fully cooperative games, we assume there exists at least one joint policy $\pi_{\theta}$ such that all agents can achieve the maximal cumulative reward with this joint policy.
\subsection{A Variational Lower Bound for Multi-Agent Reinforcement Learning Problems}
We transform the control problem into an inference problem by introducing a binary random variable $o^{i}_t$ which serves as the indicator for “optimality” for each agent $i$ at each time step $t$. Recall that in single agent problem, reward $R(s_t, a_t)$ is bounded, but the achievement of the maximum reward given the action $a_t$ is unknown. Therefore, in the single-agent case, $o_t$ indicates the optimality of achieving the bounded maximum reward $r^*_t$. It thus can be regarded as a random variable and we have $P(o_t=1|s_t, a_t)\propto \exp(R(s_t, a_t))$. Intuitively, this formulation dictates that higher rewards reflect a higher likelihood of achieving optimality, i.e., the case when $o_{t}=1$. However, the definition of ``optimality" in the multi-agent case is subtlety different from the one in the single-agent situation. 

In cooperative multi-agent reinforcement learning (CMARL), to define agent $i$'s optimality, we first introduce the definition of optimum and optimal policy:
\begin{definition}
\label{def:optimum}
    In cooperative multi-agent reinforcement learning, optimum is a strategy profile $(\pi^{1*}, \hdots, \pi^{n*})$ such that: 
    \begin{align}
        &\mathbb { E }_{s\sim p_s,a^{i*} _ { t }\sim \pi^{i*}, a^{-i*} _ { t  }\sim \pi^{-i*}} \left[ \sum _ { t = 1 } ^ { \infty } \gamma ^ { t } R^i( s_{t}, a^{i*} _ { t }, a^{-i*} _ { t  }  ) \right] \nonumber\\ 
        &\geq \mathbb { E }_{s\sim p_s,a^{i} _ { t }\sim \pi^{i}, a^{-i} _ { t  }\sim \pi^{-i}} \left[ \sum _ { t = 1 } ^ { \infty } \gamma ^ { t } R^i( s_{t}, a^{i} _ { t }, a^{-i} _ { t  }  ) \right] \\
        &\forall\pi\in\Pi, i\in{(1\hdots n)}\nonumber,
    \end{align}
    where $\pi=(\pi^i, \pi^{-i})$ and  Agent $i$'s optimal policy is $\pi^{i*}$. 
\end{definition}

In CMARL, a single agent's ``optimality" $o^i_t$ cannot imply that it obtains the maximum reward because the reward depends on the joint actions of all agents $(a^i, a^{-i})$. Therefore, we define $o^i_t=1$ only indicates that \emph{agent $i$'s policy at time step $t$ is optimal}. The posterior probability of agent $i$'s optimality given its action $a^i_t$ is the probability that the action is sampled from the optimal policy:
\begin{align}
    P(o^i_t|a^i_t) = P(a^i_t\sim\pi^{i*}|a^i_t)=\pi^{i*}(a^i_t).
\end{align}
We also assume that given other players playing optimally ($o^{-i}_t=1$), the higher the reward agent $i$ receives the higher the probability of agent $i$'s current policy is optimal ($o^i_t=1$):
\begin{align}
    P(o^i_t=1|o^{-i}_t=1, s_t,a^i_t, a^{-i}_t) \propto \exp(R(s_t, a^i_t, a^{-i}_t)).
\end{align}
The conditional probabilities of ``optimality" in both of CMARL and single-agent case have similar forms. However, it is worth mentioning that the ``optimality" in CMARL has a different interpretation to the one in single-agent case.

For cooperative games, if all agents play optimally, then agents can receive the maximum rewards, which is the optimum of the games. Therefore, given the fact that other agents are playing their optimal policies $o^{-i}=1$, the probability that agent $i$ also plays its optimal policy $P(o^i=1|o^{-i}=1)$ is the probability of obtaining the maximum reward from agent $i$'s perspective. Therefore, we define agent $i$'s objective as:
 \begin{equation}
     \max ~~ \mathcal{J} \overset{\Delta}{=} \log P(o^i_{1:T}=1|o^{-i}_{1:T}=1)
     \label{eq:exact_obj}
 \end{equation}
As we assume no knowledge of the optimal policies and the model of the environment, we treat them as latent variables. To optimize the observed evidence defined in Eq.~\ref{eq:exact_obj}, therefore, we use variational inference (VI) with an auxiliary distribution over these latent variables $q(a^i_{1:T}, a^{-i}_{1:T}, s_{1:T}|o^i_{1:T}=1, o^{-i}_{1:T}=1)$. Without loss of generality, we here derive the solution for agent $i$. We factorize $q(a^i_{1:T}, a^{-i}_{1:T}, s_{1:T}|o^i_{1:T}=1, o^{-i}_{1:T}=1)$ so as to capture agent $i$'s conditional policy on the current state and opponents actions, and beliefs regarding opponents actions. This way, agent $i$ will learn optimal policy, while also possessing the capability to model opponents actions $a^{-i}$. Using all modelling assumptions, we may factorize $q(a^i_{1:T}, a^{-i}_{1:T}, s_{1:T}|o^i_{1:T}=1, o^{-i}_{1:T}=1)$ as: 
\begin{align}
    &q(a^i_{1:T}, a^{-i}_{1:T}, s_{1:T}|o^i_{1:T}=1, o^{-i}_{1:T}=1) \nonumber\\
    &= P(s_1) \prod_{t}P(s_{t+1}|s_t, a_t)q(a^{i}_t|a^{-i}_t, s_t, o^{i}_t= o^{-i}_t=1)\nonumber\\
    &\times q(a^{-i}_t|s_t,o^{i}_t= o^{-i}_t=1) \nonumber \\
    &=P(s_1) \prod_{t}P(s_{t+1}|s_t, a_t) \pi(a^{i}_t|s_t, a^{-i}_t)
    \rho(a^{-i}_t|s_t), \nonumber
\end{align}
where we have assumed the same initial and states transitions as in the original model. With this factorization, we derive a lower bound on the likelihood of optimality of agent $i$:
 \begin{align}
     &\log P(o^i_{1:T}=1|o^{-i}_{1:T}=1) \nonumber \\
     &\geq \mathcal{J}(\pi, \rho) \overset{\Delta}{=}\sum_t\mathbb{E}_{(s_t, a^i_t, a^{-i}_t)\sim q}[R^i(s_t,a^i_t, a^{-i}_t) \nonumber \\
     &+H(\pi(a^i_t|s_t, a_t^{-i}))-D _ { \mathrm { KL } }(\rho(a^{-i}_t|s_t)||P(a^{-i}_t|s_t))]\label{eq:objetiveinstochasticgame} \\
     &=\sum_t\mathbb{E}_{s_t}[\underbrace{\mathbb{E}_{a^i_t\sim\pi,a^{-i}_t\sim\rho}[R^i(s_t,a^i_t, a^{-i}_t)+H(\pi(a^i_t|s_t, a_t^{-i}))}_{\text{MEO}}]\nonumber \\
     & -\underbrace{\mathbb{E}_{a^{-i}_t\sim\rho}[D _ { \mathrm { KL } }(\rho(a^{-i}_t|s_t)||P(a^{-i}_t|s_t))]}_{\text{Regularizer of }\rho}]. \label{eq:novelobjective}
 \end{align}
Written out in full, $\rho(a^{-i}_t|s_t, o^{-i}_t=1)$ is agent $i$'s opponent model estimating \emph{optimal} policies of its opponents, $\pi(a^i_t|s_t, a^{-i}_t, o^i_t=1, o^{-i}_t=1)$ is the agent $i$'s conditional policy at optimum ($o^i_t=o^{-i}_t=1$) and $P(a^{-i}_t|s_t, o^{-i}_t=1)$ is the prior of optimal policy of opponents. In our work, we set the prior $P(a^{-i}_t|s_t, o^{-i}_t=1)$ equal to the observed empirical distribution of opponents' actions given states. As we are only interested in the case where $(o^i_t=1, o^{-i}_t=1)$, we drop them in $\pi, \rho$ and $P(a^{-i}_t|s_t)$ here and thereafter. $H(\cdot)$ is the entropy function. Eq.~\ref{eq:objetiveinstochasticgame} is a variational lower bound of $\log P(o^i_{1:T}=1|o^{-i}_{1:T}=1)$ and the derivation is deferred to Appendix~\ref{appendix:lowerboundofoptimality}.
 
\subsection{The Learning of Opponent Model}
\label{sec:learningofopponent}

 We can further expand Eq.~\ref{eq:objetiveinstochasticgame} into Eq.~\ref{eq:novelobjective} and we find that it resembles the maximum entropy objective in single-agent reinforcement learning~\cite{kappen2005,Todorov2007,ziebart2008maximum,HaarnojaTAL17rldep}. We denote agent $i$'s expectation of reward $R(s_t,a^i_t, a^{-i}_t)$ plus entropy of the conditional policy $H(\pi(a^i|s, a^-i))$ as agent $i$'s \textit{maximum entropy objective (MEO)}. In the multi-agent version, however, it is worthy of noting that optimizing the MEO will lead to \emph{the optimization of $\rho$}. This can be counter-intuitive at first sight as opponent behaviour models are normally trained with only past state-action data $(s, a^{-i})$  to predict opponents' actions.

However, recall that $\rho(a^{-i}_t|s_t, o^{-i}_t=1)$ is modelling opponents' \emph{optimal} policies in our work. Given agent $i$'s policy $\pi^i$ being fixed, optimizing MEO with respect to $\rho$ updates agent $i$'s opponent model in the direction of the higher shared reward $R(s, a^i, a^{-i})$ and the more stochastic conditional policy $\pi^i(a^i|s, a^{-i})$, making it closer to the real optimal policies of the opponents. Without any regularization, at iteration $d$, agent $i$ can freely learn a new opponent model $\rho^i_{d+1}$ which is the closest to the optimal opponent policies $\pi^{-i*}$ from its perspective given $\pi^i_{d}(a^i|s, a^{-i})$. Next, agent $i$ can optimize the lower bound with respect to $\pi^i_{d+1}(a^i|s, a^{-i})$ given $\rho^i_{d+1}$. Then we have an EM-like iterative training and can show it monotonically increases the probability that the opponent model $\rho$ is optimal policies of the opponents. Then, by acting optimally to the converged opponent model $\rho^{i\infty}$, we can recover agent $i$'s optimal policy $\pi^{i*}$.
 
However, it is unrealistic to learn such an opponent model. As the opponents have no access to agent $i$'s conditional policy $\pi^i_d(a^i|s, a^{-i})$, the learning of its policy can be different from the one of agent $i$'s opponent model. Then the actual opponent policies $\pi^{-i}_{d+1}$ can be very different from agent $i$'s converged opponent model $\rho^{i\infty}$ learned in the above way given agent $i$'s conditional policy $\pi^i_d({a^i|s, a^{-i}})$. Therefore, acting optimally to an opponent model far from the real opponents' policies can lead to poor performance.
 
The last term in Eq.~\ref{eq:novelobjective} can prevent agent $i$ building an unrealistic opponent model. The Kullback-Leibler (KL) divergence between opponent model and a prior $D _ { \mathrm { KL } }(\rho(a^{-i}_t|s_t)||P(a^{-i}_t|s_t))$ can act as a regularizer of $\rho$. By setting the prior to the empirical distribution of opponent past behaviour, the KL divergence penalizes $\rho$ heavily if it deviates from the empirical distribution too much. As the objective in Eq.~\ref{eq:novelobjective} can be seen as a \textit{Maximum Entropy objective} for one agent's policy and opponent model with \textit{regularization on the opponent model}, we call this objective as \textit{Regularized Opponent Model with Maximum Entropy Objective (ROMMEO)}.  


\section{Multi-Agent Soft Actor Critic} 
To optimize the ROMMEO in Eq.~\ref{eq:novelobjective} derived in the previous section, we propose two off-policy algorithms. We first introduce an exact tabular Q-iteration method with proof of convergence. For practical implementation in a complex continuous environment, we then propose the ROMMEO actor critic ROMMEO-AC, which is an approximation to this procedure.

\subsection{Regularized Opponent Model with Maximum Entropy Objective Q-Iteration}

In this section, we derive a multi-agent version of Soft Q-iteration algorithm proposed in~\cite{HaarnojaTAL17rldep} and we name our algorithm as ROMMEO-Q. The derivation follows from a similar logic to~\cite{HaarnojaTAL17rldep}, but the extension of Soft Q-learning to MARL is still nontrivial. From this section, we slightly modify the objective in Eq.~\ref{eq:novelobjective} by adding a weighting factor $\alpha$ for the entropy term and the original objective can be recovered by setting $\alpha=1$.

We first define multi-agent soft Q-function and V-function respectively. Then we can show that the conditional policy and opponent model defined in Eq.~\ref{eq:optimalpifromsoftq} and \ref{eq:optimalrhofromsoftq} below are optimal solutions with respect to the objective defined in Eq.~\ref{eq:novelobjective}:

\begin{theorem}
    \label{theorem:optimalpiandrho}
    We define the soft state-action value function of agent $i$ as
    \begin{equation}
    \small{
       \begin{aligned}
        &Q^{\pi^*,\rho^*}_{soft}(s_t, a^i_t, a^{-i}_t) = r_t+ \mathbb{E}_{(s_{t+l},a^{i}_{t+l}, a^{-i}_{t+l}, \hdots)\sim q}[\sum_{l=1}^{\infty}\gamma^{l}(r_{t+l}\\
        &  +\alpha H(\pi^*(a^{i}_{t+l}|a_{t+l}^{-i}, s_{t+l}))  -D _ { \mathrm { KL } }(\rho^*(a^{-i}_{t+l}|s_{t+l})||P(a_{t+l}^{-i}|s_{t+l}))],
    \end{aligned} }
    \end{equation}

    and soft state value function as
    \begin{equation}
    \resizebox{0.5\textwidth}{!}{$ 
        V^*(s)=\log\sum_{a^{-i}}P(a^{-i}|s)\left(\sum_{a^i}\exp(\frac{1}{\alpha}Q^*_{soft}(s, a^i, a^{-i}))\right)^{\alpha},
    $}
    \end{equation}
    Then the optimal conditional policy and opponent model for Eq.~\ref{eq:objetiveinstochasticgame} are
    \begin{flalign}
        \pi^*(a^i|s, a^{-i})= \frac{\exp(\frac{1}{\alpha}Q^{\pi^*,\rho^*}_{soft}(s, a^i,a^{-i}))}{\sum_{a^i}\exp(\frac{1}{\alpha}Q^{\pi^*,\rho^*}_{soft}(s, a^i,a^{-i}))} \label{eq:optimalpifromsoftq},
    \end{flalign}
    and 
    \begin{align}
        \rho^*(a^{-i}|s)= \frac{P(a^{-i}|s)\left(\sum_{a^i}\exp(\frac{1}{\alpha}Q^*_{soft}(s, a^i, a^{-i}))\right)^{\alpha}}{\exp(V^*(s))} \label{eq:optimalrhofromsoftq}.
    \end{align}
\end{theorem}
\begin{proof}
    See Appendix~\ref{appendix:policyimprovementandopponentmodelimprovement}.
\end{proof}

Following from Theorem \ref{theorem:optimalpiandrho}, we can find the optimal solution of Eq.~\ref{eq:novelobjective} by learning the soft multi-agent Q-function first and recover the optimal policy $\pi^*$ and opponent model $\rho^*$ by Equations \ref{eq:optimalpifromsoftq} and \ref{eq:optimalrhofromsoftq}. To learn the Q-function, we show that it satisfies a Bellman-like equation, which we name it as multi-agent soft Bellman equation:
\begin{theorem}
    \label{theorem:qiteration}
    We define the soft multi-agent Bellman equation for the soft state-action value function $Q^{\pi, \rho}_{soft}(s, a^i, a^{-i})$ of agent $i$ as
    \begin{align}
        Q^{\pi,\rho}_{soft}(s, a^i, a^{-i}) = r_t + \gamma\mathbb{E}_{(s_{t+1})}[V_{soft}(s_{t+1})].
        \label{eq:multiagentsoftbellmaneq}
    \end{align}
\end{theorem}
\begin{proof}
    See Appendix~\ref{appendix:softbellmanequation}.
\end{proof}

With this Bellman equation defined above, we can derive a solution to Eq.~\ref{eq:multiagentsoftbellmaneq} with a fixed point iteration, which we call ROMMEO Q-iteration (ROMMEO-Q). Additionally, We can show that it can converge to the optimal $Q^*_{soft}$ and $V^*_{soft}$ with certain restrictions as stated in~\cite{wen2018probabilistic}:
\begin{theorem}
    \label{theorem:convergence}
    ROMMEO Q-iteration. In a symmetric game with only one global optimum, i.e. $\mathbb { E } _ { \pi ^ { * } } \left[ Q _ { t } ^ { i } ( s ) \right] \geq \mathbb { E } _ { \pi } \left[ Q _ { t } ^ { i } ( s ) \right]$, where $\pi^*$ is the optimal strategy profile. Let $Q_{soft}(\cdot, \cdot, \cdot)$ and $V_{soft}(\cdot)$ be bounded and assume $$\sum_{a^{-i}}P(a^{-i}|s)\left(\sum_{a^i}\exp(\frac{1}{\alpha}Q^*_{soft}(s, a^i, a^{-i}))\right)^{\alpha}<\infty$$ and that $ Q^*_{soft}<\infty $ exists. Then the fixed-point iteration
    \begin{equation}
    Q_{soft}(s_t, a^i_t, a^{-i}_t)\leftarrow r_t+\gamma\mathbb{E}_{(s_{t+1})}[V_{soft}(s_{t+1})], 
    \label{eq:qierationupdate}
    \end{equation}
    where
    $
    V_{soft}(s_t)\leftarrow \log\sum_{a^{-i}_t}P(a^{-i}_t|s_t)
    \times\left(\sum_{a^i_t}\exp(\frac{1}{\alpha}Q_{soft}(s_t, a^i_t, a^{-i}_t))\right)^{\alpha}
    $
    ,
    $\forall{s_t, a^i_t, a^{-i}_t},$ converges to $Q^*_{soft}$ and $V^*_{soft}$ respectively.
\end{theorem}
\begin{proof}
    See Appendix~\ref{appendix:softbellmanequation}.
\end{proof}
\begin{figure*}[t!]
  \centering
  \subcaptionbox{\label{fig:lc_icg}
  }
  {
  \includegraphics[height=8\baselineskip,width=0.3\textwidth]{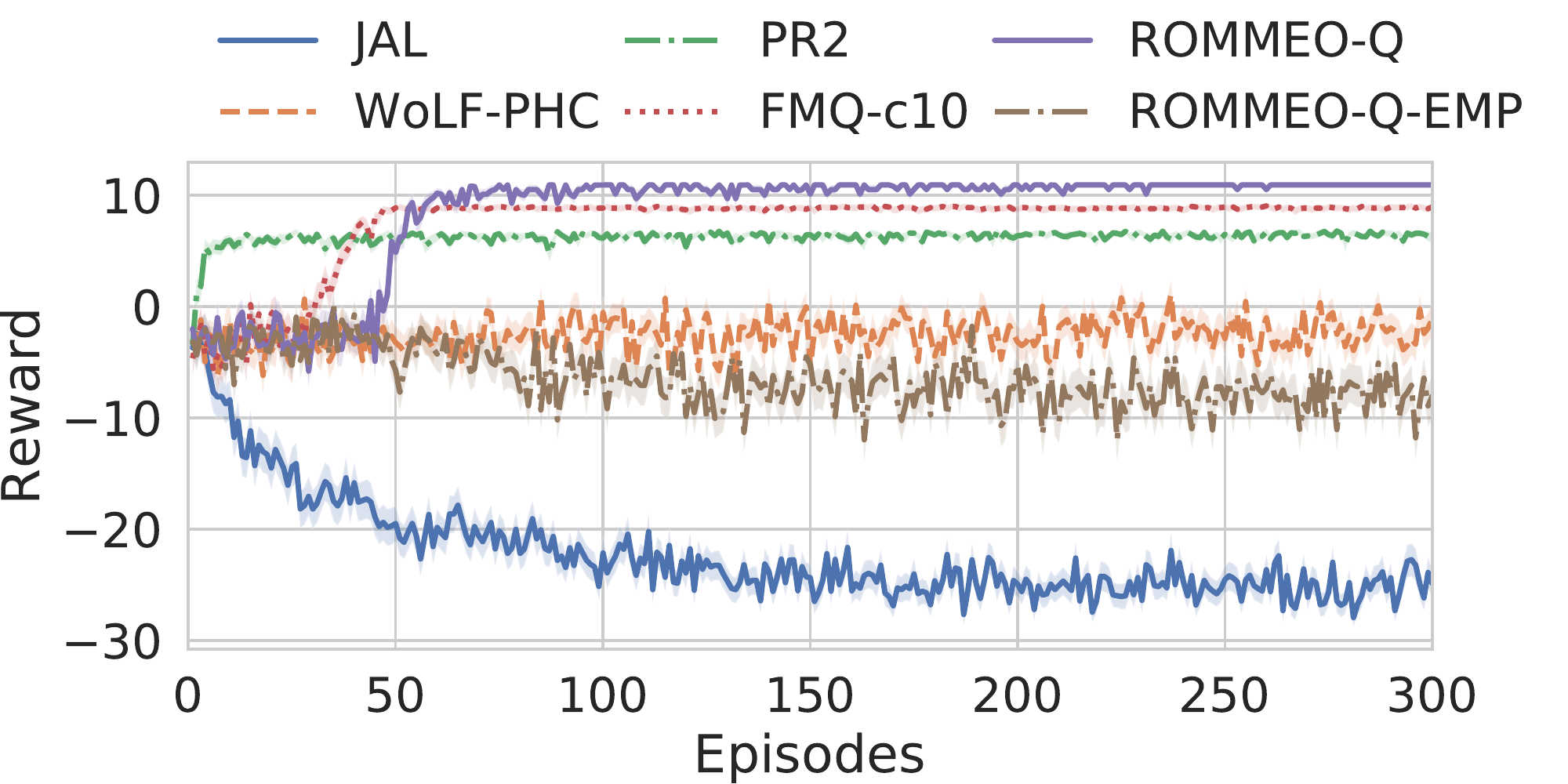}} \quad
  \subcaptionbox{
  \label{fig:pc_icg}
  }
  {
  \includegraphics[height=8\baselineskip,width=0.3\textwidth]{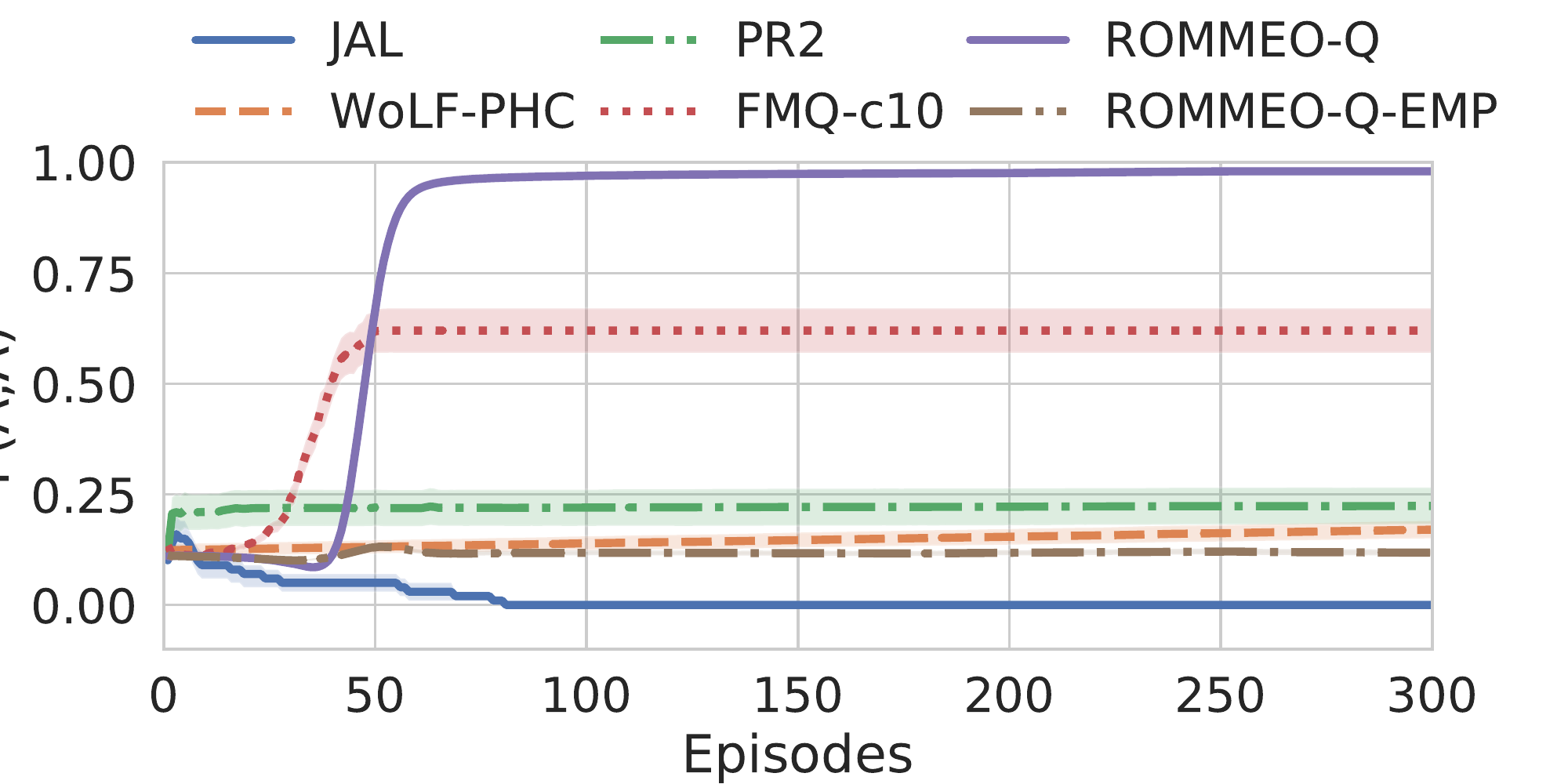}} \quad
  \subcaptionbox{
  \label{fig:rho_vs_emp}}
  {
  \includegraphics[height=8\baselineskip,width=0.3\textwidth]{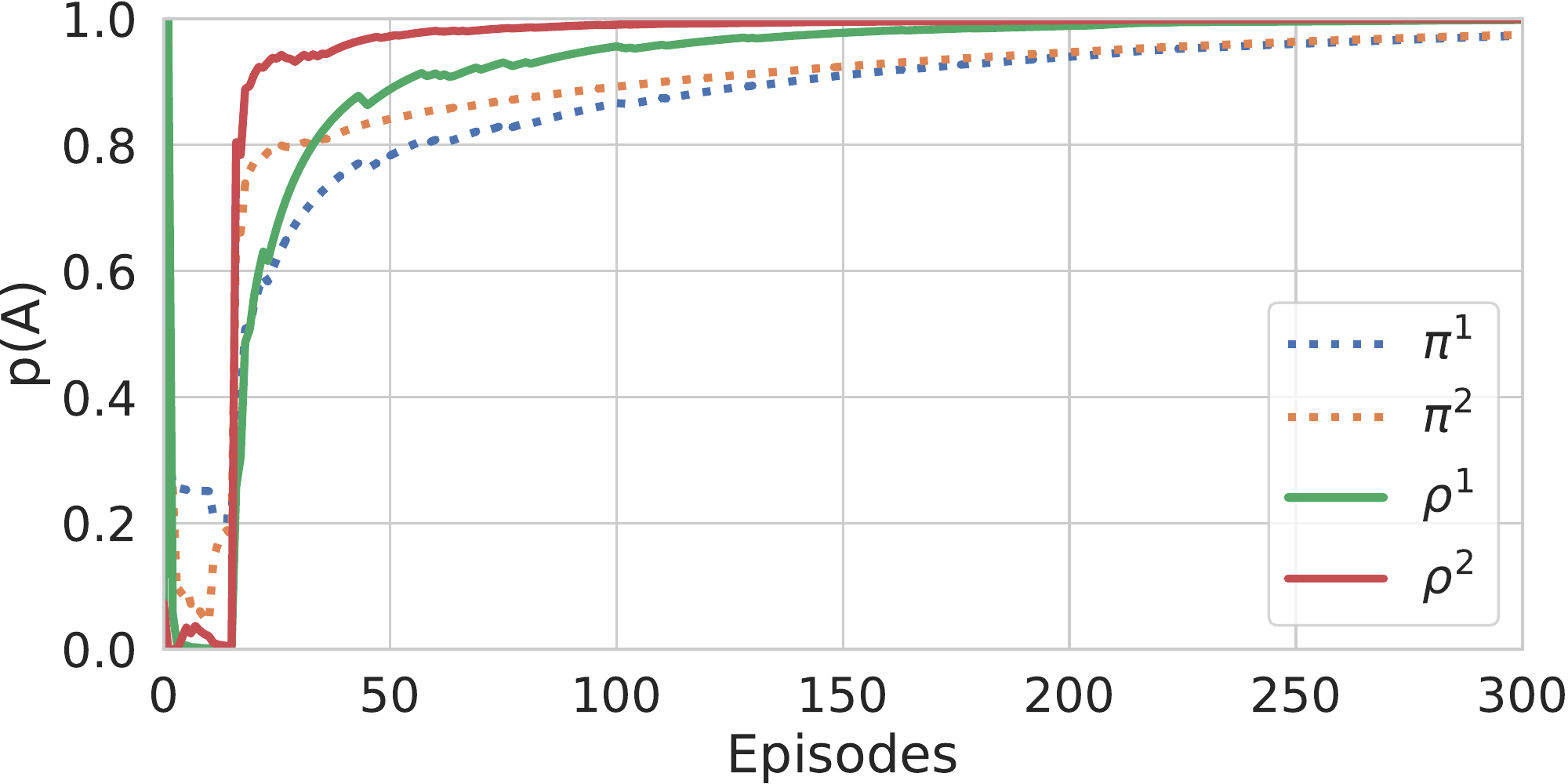}} \quad
  \caption{(a): Learning curves of ROMMEO and baselines on ICG over 100 episodes. (b): Probability of convergence to the global optimum for ROMMEO and baselines on ICG over 100 episodes. The vertical axis is the joint probability of taking actions $A$ for both agents. (c): Probability of taking $A$ estimated by agent $i$'s opponent model $\rho^i$ and observed empirical frequency $P^i$ in one trail of training, ${i\in\{1, 2 \}}$.}
  \label{fig:tabularq}
\end{figure*}

\subsection{Regularized Opponent Model with Maximum Entropy Objective Actor Critic}
The ROMMEO-Q assumes we have the model of the environment and is impractical to implement in high-dimensional continuous problems. To solve these problems, we propose the ROMMEO actor critic (ROMMEO-AC) which is a model-free method. We use neural networks (NNs) as function approximators for the conditional policy, opponent model and Q-function and learn these functions by stochastic gradient. We parameterize the Q-function, conditional policy and opponent model by $Q_{\omega}(s, a^i, a^{-i})$, $\pi_{\theta}(a^i_t|s_t, a^{-i}_t)$ and $\rho_{\phi}(a^{-i}_t|s_t)$ respectively.

Without access to the environment model, we first replace the Q-iteration with Q-learning. Therefore, we can train $\omega$ to minimize:
\begin{align}
    \mathcal{J}_{Q}(\omega) & =\mathbb{E}_{(s_t, a^i_t, a^{-i}_t)\sim \mathcal{D}}[\frac{1}{2}(Q_{\omega}(s_t, a^i_t, a^{-i}_t) \nonumber\\
    &- R(s_t, a^i_t, a^{-i}_t) - \gamma\mathbb{E}_{s_{t+1}\sim p_s}[\Bar{V}(s_{t+1})])^2],
    \label{eq:softac-omega-obj}
\end{align}
with
\begin{align}
    \Bar{V}(s_{t+1})&=Q_{\Bar{\omega}}(s_{t+1}, a^{i}_{t+1}, \hat{a}^{-i}_{t+1})-\log\rho_{\phi}(\hat{a}^{-i}_{t+1}|s_{t+1})\nonumber\\
    &-\alpha \log\pi_{\theta}(a^{i}_{t+1}|s_{t+1}, \hat{a}^{-i}_{t+1})+logP(\hat{a}^{-i}_{t+1}|s_{t+1}),
    \label{eq:target_v_values}
\end{align}
where $Q_{\Bar{\omega}}$ are target functions for providing relatively stable target values. We use $\hat{a}^{-i}_t$ denoting the action sampled from agent $i$'s opponent model $\rho(a^{-i}_t|s_t)$ and it should be distinguished from $a^{-i}_t$ which is the real action taken by agent $i$'s opponent. Eq.~\ref{eq:target_v_values} can be derived from Eq.~\ref{eq:optimalpifromsoftq} and \ref{eq:optimalrhofromsoftq}.

To recover the optimal conditional policy and opponent model and avoid intractable inference steps defined in Eq.~\ref{eq:optimalpifromsoftq} and \ref{eq:optimalrhofromsoftq} in complex problems, we follow the method in~\cite{haarnoja2018softactorcritic} where $\theta$ and $\phi$ are trained to minimize the KL-divergence:
\begin{align}
    &\mathcal{J}_{\pi}(\theta)=\mathbb{E}_{s_t\sim D, a^{-i}_t\sim\rho} \nonumber\\
    &\left [  D _ { \mathrm { KL } }\left(\pi_{\theta} (\cdot|s_t, \hat{a}^{-i}_t)\bigg|\bigg|\frac{\exp(\frac{1}{\alpha}Q_{\omega}(s_t, \cdot, \hat{a}^{-i}_t))}{Z_{\omega}(s_t, \hat{a}^{-i}_t)}\right ) \right ],
\end{align}
\begin{align}
    &\mathcal{J}_{\rho}(\phi)=\mathbb{E}_{(s_t, a^i_t)\sim D}\nonumber\\
    &\left[D _ { \mathrm { KL } }\left ( \rho(\cdot|s_t)\bigg|\bigg|\frac{P(\cdot|s_t)\left(\frac{\exp(\frac{1}{\alpha}Q(s_t,a^i_t,\cdot))}{\pi_{\theta}(a^i_t|s_t, \cdot)}\right)^{\alpha}}{Z_{\omega}(s_t)} \right) \right].
\end{align}
By using the reparameterization trick: $\hat{a}^{-i}_t = g_{\phi}(\epsilon^{-i}_t;s_t)$ and $a^i_t = f_{\theta}(\epsilon^i_t;s_t, \hat{a}^{-i}_t)$, we can rewrite the objectives above as: 
\begin{align}
    \mathcal{J}_{\pi}(\theta)&=\mathbb{E}_{s_t\sim D, \epsilon^i_t\sim N, \hat{a}^{-i}_t\sim \rho}[\alpha \log\pi_{\theta}(f_{\theta}(\epsilon^i_t;s_t,\hat{a}^{-i}_t))\nonumber\\
    &-Q_{\omega}(s_t,f_{\theta}(\epsilon^i_t;s_t, \hat{a}^{-i}_t), \hat{a}^{-i}_t)],
    \label{eq:softac-theta-obj}
\end{align}
\begin{align}
    \mathcal{J}_{\rho}(\phi)&=\mathbb{E}_{(s_t, a_t)\sim D, \epsilon^{-i}_t\sim N}[\log\rho_{\phi}(g_{\phi}(\epsilon^{-i}_t;s_t)|s_t)\nonumber\\
    &-logP(\hat{a}^{-i}_t|s_t)-Q(s_t, a^i_t, g_{\phi}(\epsilon^{-i}_t;s_t)) \nonumber\\
    &+\alpha \log\pi_{\theta}(a^i_t|s_t, g_{\phi}(\epsilon^{-i}_t;s_t))].
    \label{eq:softac-phi-obj}
\end{align}

The gradient of Eq.~\ref{eq:softac-omega-obj}, \ref{eq:softac-theta-obj} and \ref{eq:softac-phi-obj} with respect to the corresponding parameters are listed as below:
\begin{align}
    \nabla_{\omega}\mathcal{J}_Q (\omega) 
    &=\nabla_{\omega}Q_{\omega}(s_t, a^{i}_t, a^{-i}_t)(Q_{\omega}(s_t, a^{i}_t, a^{-i}_t)\nonumber\\
    &-R(s_t, a^i_t, a^{-i}_t) - \gamma\Bar{V}(s_{t+1})),
\end{align}
\begin{align}
    \nabla_{\theta}\mathcal{J}_{\pi} (\theta)&=\nabla_{\theta}\alpha \log \pi_{\theta}(a^{i}_t|s_t,\hat{a}^{-i}_t)+\nabla_{\theta}f_{\theta}(\epsilon^i_t;s_t,\hat{a}^{-i}_t)\nonumber\\
    &(\nabla_{a^{i}_t}\alpha \log\pi_{\theta}(a^{i}_t|s_t,\hat{a}^{-i}_t) -\nabla_{a^{i}_t}Q_{\omega}(s_t,a^{i}_t,\hat{a}^{-i}_t)),
\end{align}
\begin{align}
    \nabla_{\phi}\mathcal{J}_{\rho} (\phi)&=\nabla_{\phi} \log \rho_{\phi}(\hat{a}^{-i}_t|s_t)+(\nabla_{\hat{a}^{-i}_t} \log\rho_{\phi^i}(\hat{a}^{-i}_t|s_t) \nonumber\\
    &- \nabla_{\hat{a}^{-i}_t} \log P(\hat{a}^{-i}_t|s_t) - \nabla_{\hat{a}^{-i}_t}Q_{\omega^i}(s_t,a_t^{i},\hat{a}^{-i}_t)\nonumber\\
    &+\nabla_{\hat{a}^{-i}_t}\alpha \log\pi_{\theta}(a^{i}|s_t, \hat{a}^{-i}_t)) \nabla_{\phi}g_{\phi}(\epsilon^{-i}_t;s_t).
\end{align}

We list the pseudo-code of ROMMEO-Q and ROMMEO-AC in Appendix \ref{appendix:algos}.


\section{Related Works}
In early works, the maximum entropy principle has been used in policy search in linear dynamics~\cite{Todorov2010linearly-solvable,Toussaint:2009,Levine2013} and path integral control in general dynamics~\cite{kappen2005,Theodorou:2010}. Recently, off-policy methods~\cite{HaarnojaTAL17rldep,Schulman2017equivalence,Nachum2017gap} have been proposed to improve the sample efficiency in optimizing MEO. To avoid complex sampling procedure, training a policy in supervised fashion is employed in~\cite{haarnoja2018softactorcritic}. Our work is closely related to this series of recent works because ROMMEO is an extension of MEO to MARL. 

A few related works to ours have been conducted in multi-agent soft Q-learning~\cite{wen2018probabilistic,wei2018multiagent,Grau2018balancing}, where variants of soft Q-learning are applied for solving different problems in MARL. However, unlike previous works, we do not take the soft Q-learning as given and apply it to MARL problems with modifications. In our work, we first establish a novel objective ROMMEO and ROMMEO-Q is only an off-policy method we derive with complete convergence proof, which can optimize the  objective. There are other ways of optimizing ROMMEO, for example, the on-policy gradient-based methods, but they are not included in the paper. 

There has been substantial progress in combining RL with probabilistic inference. However, most of the existing works focus on the single-agent case. The literature of Bayesian methods in MARL is limited. Among these are methods performing on cooperative games with prior knowledge on distributions of the game model and the possible strategies of others~\cite{Chalkiadakis:2003} or policy parameters and possible roles of other agents~\cite{Wilson:2010}. In our work, we assume very limited prior knowledge of the environment model, optimal policy, opponents or the observations during the play. In addition, our algorithms are fully \textit{decentralized at training and execution}, which is more challenging than problems from the centralized training~\cite{Foerster2017CMOA,lowe2017MADDPG,Rashid2018qmix}.

In our work, we give a new definition of optimality in CMARL and derive a novel objective ROMMEO. We provide two off-policy RL algorithms for optimizing ROMMEO and the exact version comes with convergence proof. In addition, we provide a natural perspective on opponent modeling in coordination problems: biasing one’s opponent model towards the optimum from its perspective but regularizing it with the empirical distribution of opponent’s real behavior.

\begin{figure*}[t!]
  \centering
   \subcaptionbox{
   \label{fig:diff_reward_surface}
   }
   {
   \includegraphics[height=8\baselineskip,width=0.25\textwidth]{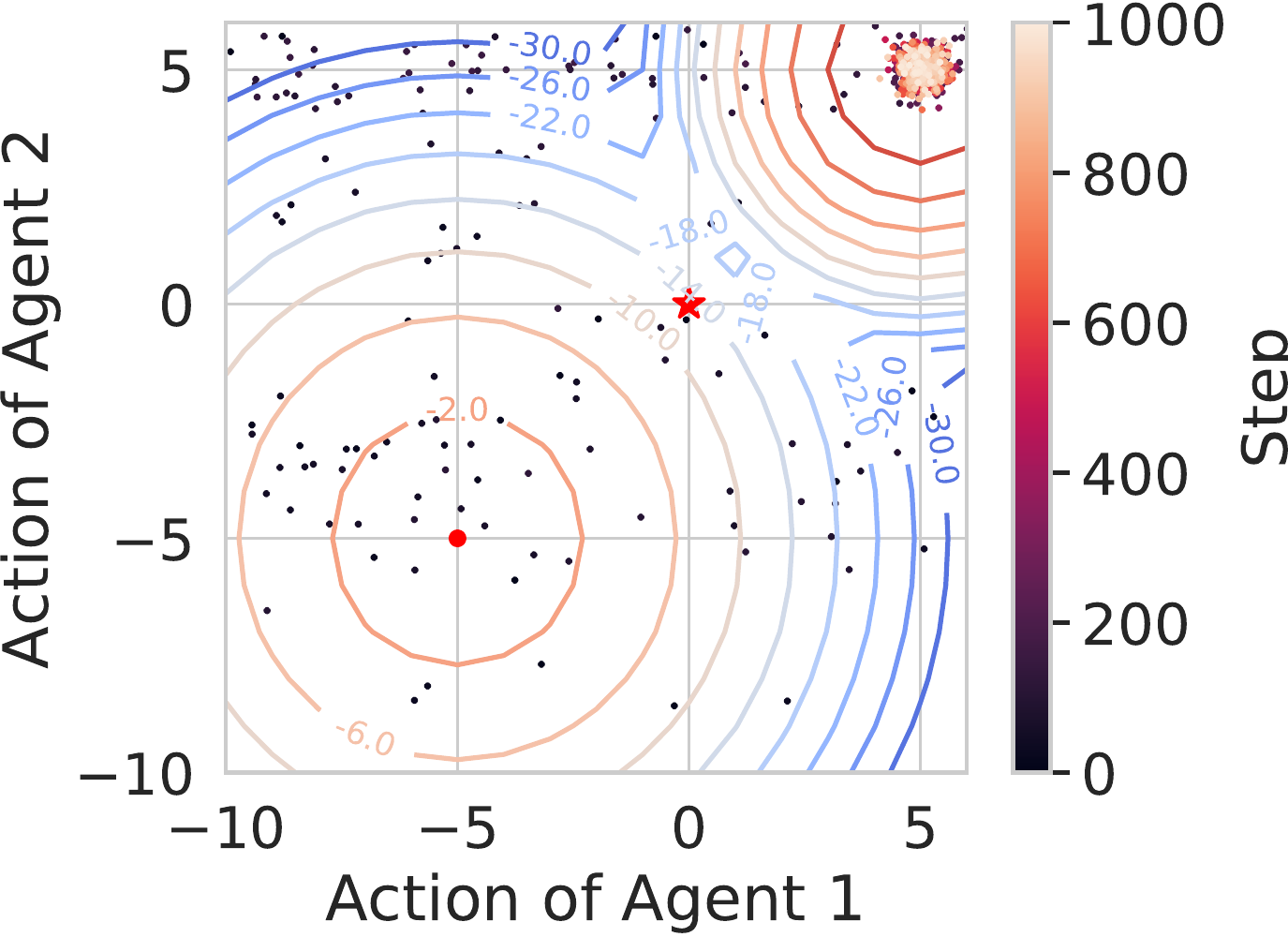}} \quad
  \subcaptionbox{
  \label{fig:diff_learning_curive}
  }
  {
  \includegraphics[height=8\baselineskip,width=0.35\textwidth]{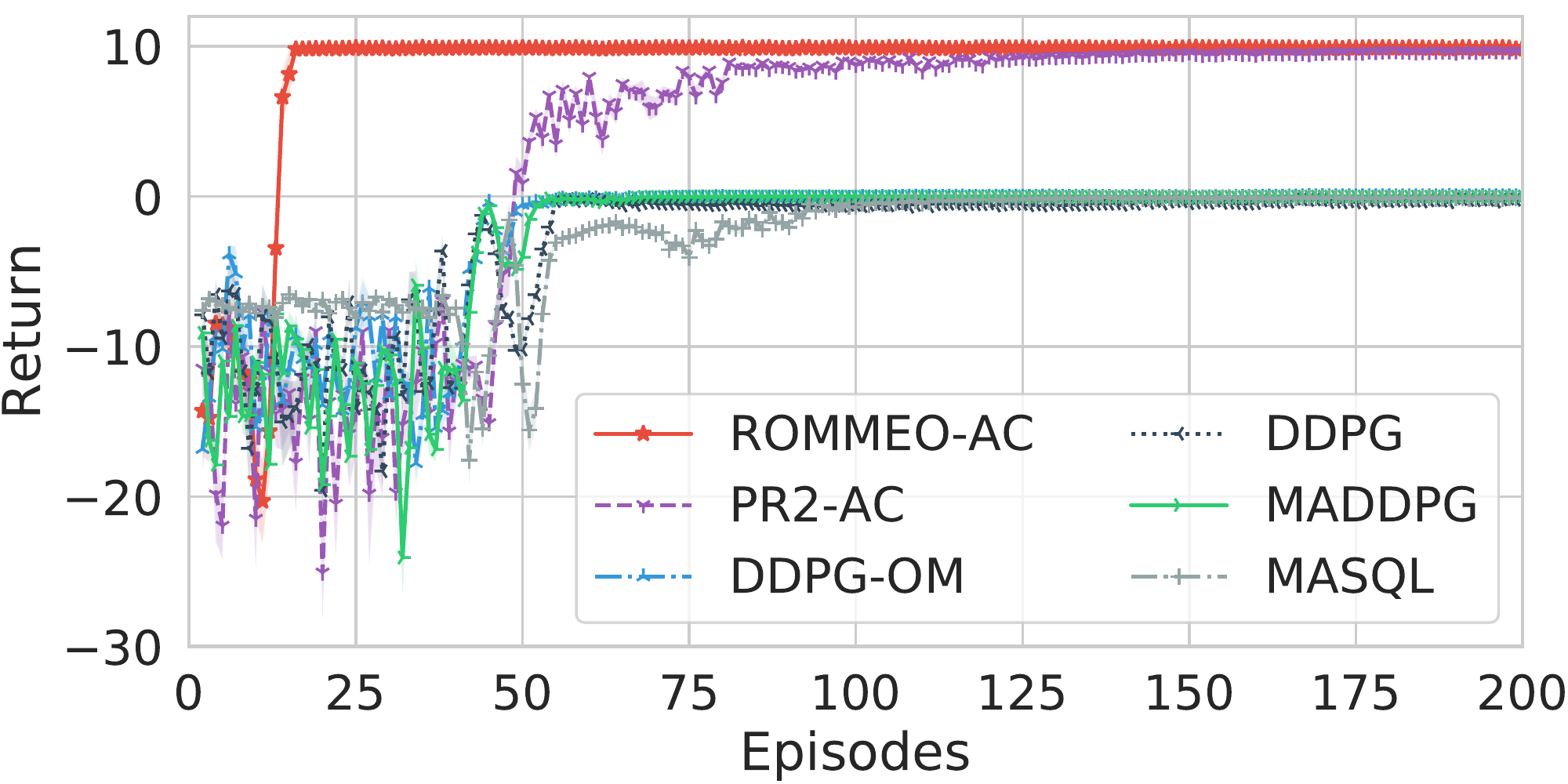}} \quad
  \subcaptionbox{
  \label{fig:diff_policies}
  }
  {
  \includegraphics[height=8\baselineskip,width=0.25\textwidth]{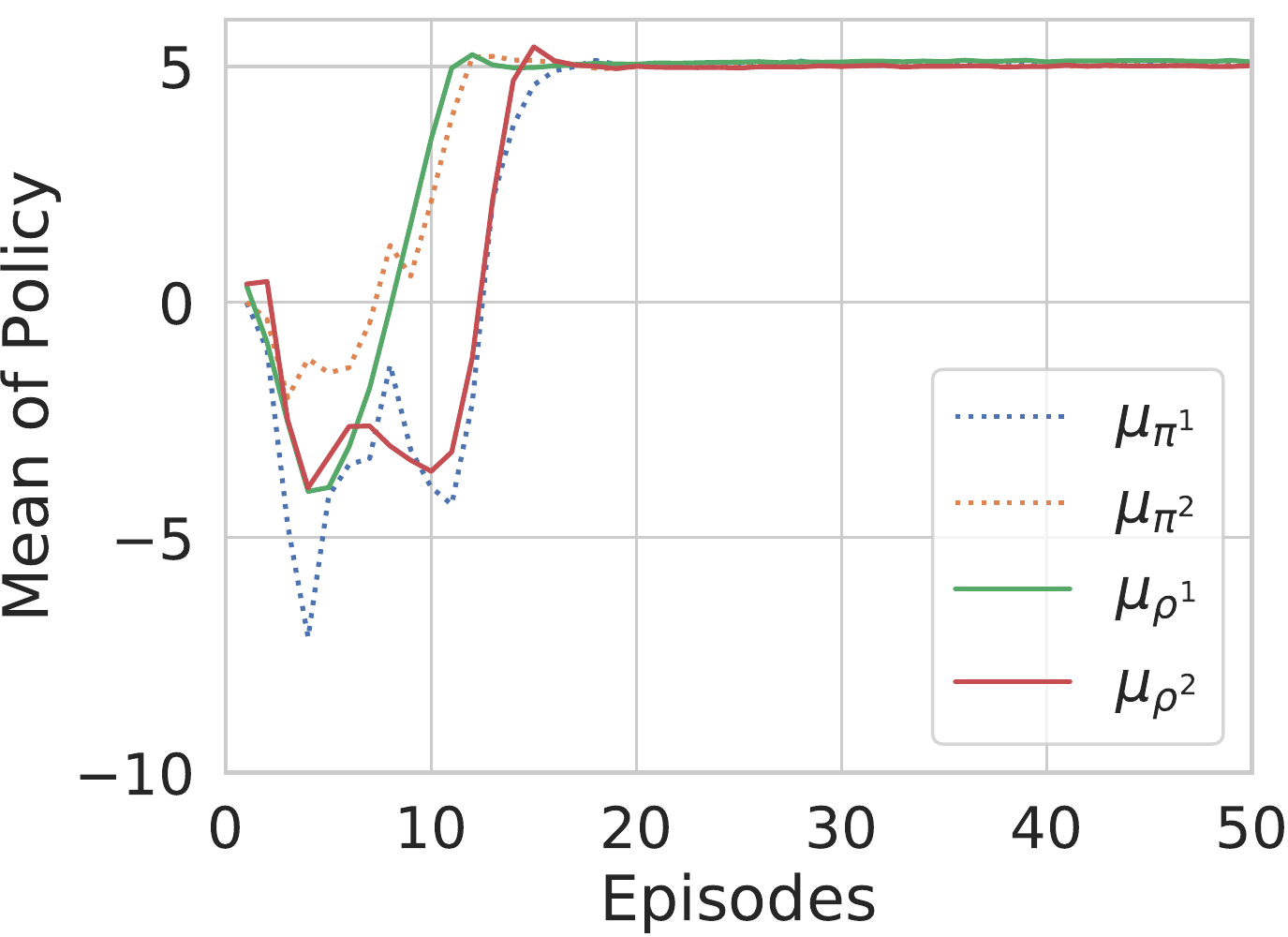}} 
  \caption{Experiment on Max of Two Quadratic Game. (a) Reward surface and learning path of agents. Scattered points are actions taken at each step; (b)  Learning curve of ROMMEO and baselines. (c) Mean of agents' policies $\pi$ and opponent models $\rho$.}
  \label{fig:diff}
\end{figure*}

\section{Experiments}
\subsection{Iterated Matrix Games}
We first present the proof-of-principle result of ROMMEO-Q\footnote{The experiment code and appendix are available at \href{https://github.com/rommeoijcai2019/rommeo}{https://github.com/rommeoijcai2019/rommeo}.} on iterated matrix games where players need to cooperate to achieve the shared maximum reward. To this end, we study the iterated climbing games (ICG) which is a classic purely cooperative two-player stateless iterated matrix games. 

Climbing game (CG) is a fully cooperative game proposed in~\cite{Claus:1998:DRL:295240.295800} whose payoff matrix is summarized as follows:
$
\small{
R = 
\begin{array}{cc}
\begin{matrix}
\end{matrix} & 
\begin{array}{rrr}
A~~~& ~~~~~~~~~~~~~~B~~~~~~~~ &~~~~~~~C~~~~ 
\end{array} \\
\begin{matrix}
A \\
B \\
C \\
\end{matrix}
 & 
 \left[
\begin{array}{lll}
{ (11 , 11)} & { (-30, -30) } & {(0,0)} \\ { (-30, -30) } & {(7, 7) } & { (6, 6) } \\ { (0, 0) } & {(0, 0) } & { (5, 3) }
\end{array}
\right]
\end{array}
}
$.
It is a challenging benchmark because of the difficulty of convergence to its global optimum. There are two Nash equilibrium $(A, A)$ and $(B, B)$ but one global optimal $(A, A)$. The punishment of miscoordination by choosing a certain action increases in the order of $C\rightarrow B\rightarrow A$. The safest action is $C$ and the miscoordination punishment is the most severe for $A$. Therefore it is very difficult for agents to converge to the global optimum in ICG.



We compare our method to a series of strong baselines in MARL, including Joint Action Learner (JAL)~\cite{Claus:1998:DRL:295240.295800}, WoLF Policy Hillclimbing (WoLF-PHC)~\cite{Bowling:2001:RCL:1642194.1642231}, Frequency Maximum Q (FMQ)~\cite{Kapetanakis:2002:RLC:777092.777145} and Probabilistic Recursive Reasoning (PR2)~\cite{wen2018probabilistic}. ROMMEO-Q-EMP is an ablation study to evaluate the effectiveness of our proposed opponent model learning process, where we replace our opponent model with empirical frequency. Fig. \ref{fig:lc_icg} shows the learning curves on ICG for different algorithms. The difference of rewards between ROMMEO-Q and FMQ-c10 may seem small because of the small reward margin between the global optimum and the local one. However, ROMMEO-Q actually outperforms all baselines significantly in terms of converging to the global optimum, which is shown in Fig. \ref{fig:pc_icg}. To further analyze the opponent modeling described in Sec. \ref{sec:learningofopponent}, we visualize the probability of agent $-i$ taking the optimal action $A$ estimated by agent $i$'s opponent model $\rho^i$ and its true policy $\pi^{-i}$ in Fig. \ref{fig:rho_vs_emp}. Agent $i$'s opponent model ``thinks ahead" of agent $-i$ and converges to agent $-i$'s optimal policy before agent $-i$ itself converges to the optimal policy. This helps agent $i$ to respond to its opponent model optimally by choosing action $A$, which in turn leads to the improvement of agent $-i$'s opponent model and policy. Therefore, the game converges to the global optimum. To note, the big drop of $P(A)$ for both policies and opponent models at the beginning of the training comes from the severe punishment of miscoordination associated with action $A$.
\subsection{Differential Games}
We adopt the differential Max of Two Quadratic Game~\cite{wei2018multiagent} for continuous case. The agents have continuous action space of $[-10, 10]$. Each agent's reward depends on the joint action following the equations:  $r ^ 1 \left( a ^ { 1 } , a ^ { 2 } \right)  =   r ^ 2 \left( a ^ { 1 } , a ^ { 2 } \right) = \max \left( f _ { 1 } , f _ { 2 } \right),$ where
$
f _ { 1 } = 0.8 \times [ -( \frac { a ^ { 1 } + 5 } { 3 } ) ^ { 2 } - ( \frac { a ^ { 2 } + 5 } { 3  } ) ^ { 2 } ], 
f _ { 2 } = 1.0 \times [ - ( \frac { a ^ { 1 } - 5 } { 1 } ) ^ { 2 } - ( \frac { a ^ { 2 } - 5 } { 1 } ) ^ { 2 } ] + 10   
$. 
We compare the algorithm with a series of baselines including PR2~\cite{wen2018probabilistic}, MASQL ~\cite{wei2018multiagent,Grau2018balancing}, MADDPG~\cite{lowe2017MADDPG} and independent learner via DDPG~\cite{lillicrap2015continuous}.
To compare against traditional opponent modeling methods, similar to  ~\cite{rabinowitz2018machine,He2016}, 
we implement an additional baseline of DDPG with an opponent module that is trained online with supervision in order to capture the latest opponent behaviors, called DDPG-OM. We trained all agents for 200 episodes with 25 steps per episode.

This is a challenging task to most continuous gradient based RL algorithms because gradient update tends to direct the training agent to the sub-optimal point.
The reward surface is provided in Fig. \ref{fig:diff_reward_surface} ; there is a local maximum $0$ at $(-5,-5)$ and a global maximum $10$ at $(5, 5)$, with a deep valley staying in the middle. If the agents' policies are initialized to $(0,0)$ (the red starred point) that lies within the basin of the left local maximum, the gradient-based methods would tend to fail to find the global maximum equilibrium point due to the valley blocking the upper right area. 

A learning path of ROMMEO-AC is summarized in Fig. \ref{fig:diff_reward_surface} and the solid bright circle on the right corner implies the convergence to the global optimum. The learning curve is presented in Fig. \ref{fig:diff_learning_curive}, ROMMEO-AC shows the capability of converging to the global optimum in a limited amount of steps, while most of the baselines can only reach the sub-optimal point. PR2-AC can also achieve the global optimum but requires many more steps to explore and learn. Additionally, fine tuning on the exploration noise or separate exploration stage is required for deterministic RL methods (MADDPG, DDPG, DDPG-OM, PR2-AC), and the learning outcomes of energy-based RL method (MASQL) is extremely sensitive to the annealing scheme for the temperature. In contrast, ROMMEO-AC employs a stochastic policy and controls the exploration level by the weighting factor $\alpha$. It does not need a separate exploration stage at the beginning of the training or a delicately designed annealing scheme for $\alpha$.

Furthermore, we analyze the learning path of policy $\pi$ and modeled opponent policy $\rho$ during the training, the results are shown in Fig.~\ref{fig:diff_policies}. The red and orange lines are mean of modeled opponent policy $\rho$, which always learn to approach the optimal ahead of the policy $\pi$ (in dashed blue and green lines). This helps the agents to establish the trust and converge to the optimum quickly, which further justifies the effectiveness and benefits of conducting a regularized opponent model proposed in Sec.~\ref{sec:learningofopponent}.
\section{Conclusion}
In this paper, we use Bayesian inference to formulate MARL problem and derive a novel objective ROMMEO which gives rise to a new perspective on opponent modeling. We design an off-policy algorithm ROMMEO-Q with complete convergence proof for optimizing ROMMEO. For better generality, we also propose ROMMEO-AC, an actor critic algorithm powered by NNs to solve complex and continuous problems. We give an insightful analysis of the effect of the new learning process of the opponent modeling on agent's performance in MARL. We evaluate our methods on the challenging matrix game and differential game and show that they can outperform a series of strong base lines. It is worthy of noting that Theorems \ref{theorem:optimalpiandrho},\ref{theorem:qiteration} and \ref{theorem:convergence} only guarantees the convergence to optimal solutions with respect to ROMMEO objective but not the optimum in the game. The achievement of the optimum in the game also relies on the opponent learning algorithm. In our work, we demonstrate that ROMMEO-Q/AC's convergence to the optimum of the game in self-play setting. The convergence to optimum of games in non-self-play settings will be studied in our future work.

\bibliographystyle{named}
\begingroup
\setstretch{.88}
{\small
\bibliography{main}
}
\endgroup
\onecolumn

\appendix

\section{Algorithms}
\label{appendix:algos}
\begin{algorithm}[H]
   \caption{Multi-agent Soft Q-learning}
   \label{alg:multiagentsoftqiteration}
\begin{algorithmic}
    \STATE {\bfseries Result:} policy $\pi^i$, opponent model $\rho^i$
    \STATE
    
    \STATE{\bfseries Initialization}:
    
    \STATE Initialize replay buffer $\mathcal{M}$ to capacity $M$.
    
    \STATE Initialize $Q_{\omega^i}(s, a^i, a^{-i})$ with random parameters $\omega^i$, $P(a^{-i}|s)$ arbitrarily, set $\gamma$ as the discount factor.
    
    \STATE Initialize target $Q_{\Bar{\omega}^i}(s, a^i, a^{-i})$ with random parameters $\Bar{\omega}^i$, set $C$ the target parameters update interval.
    \STATE
    
    \WHILE{not converge}
    \STATE
        \STATE {\bfseries Collect experience}
        \STATE
            
            \STATE For the current state $s_t$ compute the opponent model $\rho^i(a^{-i}_t|s_t)$ and conditional policy $\pi^i(a^i_t|s_t, a^{-i}_t)$ respectively from:
            $$\rho^i(a^{-i}_t|s_t)\propto P(a^{-i}_t|s_t)\left(\sum_{a^i_t}\exp (\frac{1}{\alpha} Q_{\omega^i}(s_t, a^i_t, a^{-i}_t))\right)^\alpha,$$
            
            $$\pi^i(a^i_t|s_t, \Hat{a}^{-i}_t)\propto \exp (\frac{1}{\alpha} Q_{\omega^i}(s_t, a^i_t,\Hat{a}^{-i}_t)).$$
            
            \STATE Compute the marginal policy $\pi^i(a^i_t|s_t)$ and sample an action from it:
            $$a^i_t\sim \pi^i(a^i_t|s_t)=\sum_{a^{-i}}\pi^i(a^i_t|s_t, a^{-i}_t)\rho(a^{-i}_t|s_t).$$
            
            \STATE Observe next state $s_{t+1}$, opponent action $a^{-i}_t$ and reward $r^i_t$, save the new experience in the reply buffer:
            $$\mathcal{M}\leftarrow\mathcal{M}\cup \{(s_t, a^i_t, a^{-i}_t, s_{t+1}, r^i_t)\}.$$
            
            \STATE Update the prior from the replay buffer:
            $$P(a^{-i}_t|s_t)=\frac{\sum_{m=1}^{|\mathcal{M}|}\mathbb{I}(s=s_t, a^{-i}=a^{-i}_t)}{\sum_{m=1}^{|\mathcal{M}|}\mathbb{I}(s=s_t)} \, \forall{s_t, a^{-i}_t \in \mathcal{M}}.$$
            
        \STATE{\bfseries Sample a mini-batch from the replay buffer}:
        \STATE
        $$\{s^{(n)}_t, a^{i, {(n)}}_t, a^{-i,{(n)}}_t, s^{(n)}_{t+1}, r^{(n)}_{t}\}^{N}_{n=1} \sim \mathcal{M}.$$
        
        \STATE{\bfseries Update $Q_{\omega^i}(s, a^i, a^{-i})$}:
        \STATE
        
        \FOR{each tuple $(s^{(n)}_t, a^{i, {(n)}}_t, a^{-i,{(n)}}_t, s^{(n)}_{t+1}, r^{(n)}_{t})$ }
        \STATE
        
        \STATE Sample $\{a^{-i, (n,k)}\}^K_{k=1} \sim \rho, \, \{a^{i, (n, k)}\}^K_{k=1} \sim \pi$.
        
        \STATE Compute empirical $\Bar{V}^i(s^{(n)}_{t+1})$ as:
        $$\Bar{V}^i(s^{(n)}_{t+1}) = \log \left(\frac{1}{K}\sum^K_{k=1}\frac{\left(P^{\frac{1}{\alpha}}(a^{-i, (n,k)}|s^{(n)}_{t+1})\exp (\frac{1}{\alpha}Q_{\Bar{\omega}^i}(s^{(n)}_{t+1}, a^{i, (n, k)},a^{-i, (n,k)}))\right)^\alpha}{\pi(a^{i,(n,k)}|s^{(n)}_{t+1}, a^{-i, (n, k)})\rho(a^{-i,(n,k)}|s^{(n)}_{t+1})}\right).$$
        
        \STATE Set $$ y^{(n)} = \left\{ \begin{array} { l l } { r^{(n)}_t } & { \text { for terminal } s^{(n)}_{t+1}} \\ { r^{(n)}_t + \gamma \Bar{V}^i(s^{(n)}_{t+1}) } & { \text { for non-terminal } s^{(n)}_{t+1} } \end{array} \right.$$
        
        \STATE Perform gradient descent step on $(y^{(n)}-Q_{\omega^i}(s^{(n)}_{t+1}, a^{i, (n)},a^{-i, (n)}))^2$ with respect to parameters $\omega^i$
        
        \STATE Every $C$ gradient descent steps, reset target parameters:
        $$\Bar{\omega}^i\leftarrow\omega$$
        \ENDFOR
    \ENDWHILE
    \STATE{\bfseries Compute converged $\pi^i$ and $\rho^i$}
    \STATE
\end{algorithmic}
\end{algorithm}

\begin{algorithm}[H]
  \caption{Multi-agent Variational Actor Critic}
  \label{alg:multiagentsoftPG}
\begin{algorithmic}
    \STATE {\bfseries Result:} policy $\pi_{\theta^i}$, opponent model $\rho_{\phi^i}$
    
    \STATE{\bfseries Initialization}:
    
    \STATE Initialize parameters $\theta^i$, $\phi^i$, $\omega^i$, $\psi^i$ for each agent $i$ and the random process $\mathcal{N}$ for action exploration.
    
    \STATE Assign target parameters of joint action Q-function: $\Bar { \omega }^i  \leftarrow \omega$.

    \STATE Initialize learning rates $\lambda_{V}, \lambda_{Q}, \lambda_{\pi}, \lambda{\phi}, \alpha$, and set $\gamma$ as the discount factor.
    \FOR{Each episode $d=(1,\hdots,D)$}
    \STATE Initialize random process $\mathcal{N}$ for action exploration.
    \FOR{each time step $t$}
        \STATE For the current state $s_t$, sample an action and opponent's action using:
        \STATE $\hat{a}^{-i}_t \leftarrow g_{\phi^{-i}}(\epsilon^{-i};s_t)$, where $\epsilon^{-i}_t \sim \mathcal{N}$,
        
        \STATE $a^i_t \leftarrow f_{\theta^{i}}(\epsilon^{i};s_t, \hat{a}^{-i}_t)$, where $\epsilon^{i}_t \sim \mathcal{N}$.
        
        \STATE Observe next state $s_{t+1}$, opponent action $a^{-i}_t$ and reward $r^i_t$, save the new experience in the replay buffer:
        $$\mathcal{D}^i\leftarrow\mathcal{D}^i\cup \{(s_t, a^i_t, a^{-i}_t, \hat{a}^{-i}_t, s_{t+1}, r^i_t)\}.$$
        
        \STATE Update the prior from the replay buffer:
        $$\psi^i=\argmax \mathbb{E}_{\mathcal{D}^i}[-P(a^{-i}|s)\log P_{\psi^i}(a^{-i}|s)]$$
        
        \STATE 
        \STATE Sample a mini-batch from the reply buffer:
        $$\{s^{(n)}_t, a^{i, {(n)}}_t, a^{-i,{(n)}}_t,\hat{a}^{-i,{(n)}}_t, s^{(n)}_{t+1}, r^{(n)}_{t}\}^{N}_{n=1} \sim \mathcal{M}.$$
        
        \STATE For the state $s^{(n)}_{t+1}$, sample an action and opponent's action using:
        \STATE $\hat{a}^{-i, (n)}_{t+1} \leftarrow g_{\phi^{-i}}(\epsilon^{-i};s^{(n)}_{t+1})$, where $\epsilon^{-i}_{t+1} \sim \mathcal{N}$,
        
        \STATE $a^{i, (n)}_{t+1}  \leftarrow f_{\Bar{\theta}^{i}}(\epsilon^{i};s^{(n)}_{t+1}, \hat{a}^{-i, (n)}_{t+1})$, where $\epsilon^{i}_{t+1} \sim \mathcal{N}$.
        
        \STATE $\Bar{V}^i(s^{(n)}_{t+1}) = Q_{\Bar{\omega}}(s^{(n)}_{t+1}, a^{i, (n)}_{t+1}, \hat{a}^{-i, (n)}_{t+1})-\alpha \log \pi_{\theta^i}(a^{i, (n)}_{t+1}|s^{(n)}_{t+1}, \hat{a}^{-i, (n)}_{t+1})-\log \rho_{\phi^i}(\hat{a}^{-i, (n)}_{t+1}|s^{(n)}_{t+1})+\log P_{\psi^i}(\hat{a}^{-i, (n)}_{t+1}|s^{(n)}_{t+1})$.
        
        \STATE Set $$ y^{(n)} = \left\{ \begin{array} { l l } { r^{(n)}_t } & { \text { for terminal } s^{(n)}_{t+1}} \\ { r^{(n)}_t + \gamma \Bar{V}^i(s^{(n)}_{t+1}) } & { \text { for non-terminal } s^{(n)}_{t+1} } \end{array} \right.$$
        
        \begin{align}
            &\nabla_{\omega^i}\mathcal{J}_Q (\omega^i)  \nonumber
            =\nabla_{\omega^i}Q_{\omega^i}(s^{(n)}_t, a^{i, {(n)}}_t, a^{-i,{(n)}}_t)(Q_{\omega^i}(s^{(n)}_t, a^{i, {(n)}}_t, a^{-i,{(n)}}_t)-y^{(n)}) \nonumber
        \end{align}
        
        \begin{align}
            &\nabla_{\theta^i}\mathcal{J}_{\pi} (\theta^i)=\nabla_{\theta^i}\alpha \log  \pi_{\theta^i}(a^{i, {(n)}}_t|s^{(n)}_t,\hat{a}^{-i,{(n)}}_t)  \nonumber\\
            &+(\nabla_{a^{i, {(n)}}_t}\alpha \log \pi_{\theta^i}(a^{i, {(n)}}_t|s^{(n)}_t,\hat{a}^{-i,{(n)}}_t) - \nabla_{a^{i, {(n)}}_t}Q_{\omega}(s^{(n)}_t,a^{i, {(n)}}_t,\hat{a}^{-i,{(n)}}_t))\nabla_{\theta}f_{\theta^i}(\epsilon^i_t;s^{(n)}_t,\hat{a}^{-i,{(n)}}_t) \nonumber
        \end{align}
        
        \begin{align}
            &\nabla_{\phi^i}\mathcal{J}_{\rho} (\phi^i)=\nabla_{\phi^i}\log  \rho_{\phi^i}(\hat{a}^{-i,{(n)}}_t|s^{(n)}_t)  \nonumber \\
            &+(\nabla_{\hat{a}^{-i,{(n)}}_t} \log \rho_{\phi^i}(\hat{a}^{-i,{(n)}}_t|s^{(n)}_t) - \nabla_{\hat{a}^{-i,{(n)}}_t} \log  P(\hat{a}^{-i,{(n)}}_t|s^{(n)}_t) - \nabla_{\hat{a}^{-i,{(n)}}_t}Q_{\omega^i}(s^{(n)}_t,a_t^{i, (n)},\hat{a}^{-i,{(n)}}_t)\nonumber\\
            &+\nabla_{\hat{a}^{-i,{(n)}}_t}\alpha \log \pi_{\theta^i}(a^{i, {(n)}}|s^{(n)}_t, \hat{a}^{-i,{(n)}}_t)) \nabla_{\phi^i}g_{\phi^i}(\epsilon^{-i}_t;s^{(n)}_t) \nonumber
        \end{align}
        \STATE Update parameters:
        
        $\omega^i = \omega^i - \lambda_{Q} \nabla_{\omega^i}\mathcal{J}_Q (\omega^i)$
        
        $\theta^i = \theta^i - \lambda_{\pi}\nabla_{\theta^i}\mathcal{J}_{\pi} (\theta^i)$
        
        $\phi^i = \phi^i - \lambda_{\phi^i}\nabla_{\phi^i}\mathcal{J}_{\rho} (\phi^i)$
    \ENDFOR
    \STATE Every $C$ gradient descent steps, reset target parameters: $$\overline{\omega^i} = \beta\omega^i + (1-\beta)\overline{\omega^i}$$.
    \ENDFOR
\end{algorithmic}
\end{algorithm}

\section{Variational Lower Bounds in Multi-agent Reinforcement Learning}

\subsection{The Lower Bound of The Log Likelihood of Optimality}
\label{appendix:lowerboundofoptimality}
We can factorize $P(a^i_{1:T}, a^{-i}_{1:T}, s_{1:T}|o^{-i}_{1:T})$ as :
\begin{align}
    P(a^i_{1:T}, a^{-i}_{1:T}, s_{1:T}|o^{-i}_{1:T}) = P(s_1) \prod_{t}P(s_{t+1}|s_t, a_t)P(a^{i}_t|a^{-i}_t, s_t, o^{-i}_t)P(a^{-i}_t|s_t, o^{-i}_t),
\end{align}
where $P(a^{i}_t|a^{-i}_t, s_t, o^{-i}_t)$ is the conditional policy of agent $i$ when other agents $-i$ achieve optimality. As agent $i$ has no knowledge about rewards of other agents, we set $P(a^{i}_t|a^{-i}_t, s_t, o^{-i}_t)\propto 1$.

Analogously, we factorize $q(a^i_{1:T}, a^{-i}_{1:T}, s_{1:T}|o^i_{1:T}, o^{-i}_{1:T})$ as:
\begin{align}
    q(a^i_{1:T}, a^{-i}_{1:T}, s_{1:T}|o^i_{1:T}, o^{-i}_{1:T}) &= P(s_1) \prod_{t}P(s_{t+1}|s_t, a_t) q(a^{i}_t|a^{-i}_t, s_t, o^{i}_t, o^{-i}_t)q(a^{-i}_t|s_t,o^{i}_t, o^{-i}_t) \\
    &=P(s_1) \prod_{t}P(s_{t+1}|s_t, a_t) \pi(a^{i}_t|s_t, a^{-i}_t)
    \rho(a^{-i}_t|s_t),
\end{align}
where $\pi(a^{i}_t|a^{-i}_t, s_t)$ is agent $1$'s conditional policy at optimum and $\rho(a^{-i}_t|s_t)$ is agent $1$'s model about opponents' optimal policies.

With the above factorization, we have:
\begin{align}
    &\log P(o^i_{1:T}|o^{-i}_{1:T}) \nonumber \\
    &=\log \sum_{a^i_{1:T}, a^{-i}_{1:T}, s_{1:T}}P(o^i_{1:T},a^i_{1:T}, a^{-i}_{1:T}, s_{1:T}|o^{-i}_{1:T}) \\
    &\geq \sum q(a^i_{1:T}, a^{-i}_{1:T}, s_{1:T}|o^i_{1:T}, o^{-i}_{1:T})\log \frac{P(o^i_{1:T},a^i_{1:T}, a^{-i}_{1:T}, s_{1:T}|o^{-i}_{1:T})}{q(a^i_{1:T}, a^{-i}_{1:T}, s_{1:T}|o^i_{1:T}, o^{-i}_{1:T})} \\
    &= \mathbb{E}_{(a^i_{1:T}, a^{-i}_{1:T}, s_{1:T}\sim q)}[\sum_{t=1}^T\log P(o^i_t|o^{-i}_t,s_t, a^i_t, a^{-i}_t)+\bcancel{\log P(s_1)}+\bcancel{\sum_{t=1}^T\log P(s_{t+1}|s_t, a^i_t, a^{-i}_t)} \\
    &\bcancel{-\log P(s_1)}-\bcancel{\sum_{t=1}^T\log P(s_{t+1}|s_t, a^i_t, a^{-i}_t)} \\
    &- \sum_{t=1}^T\log \pi(a^i_t|s_t, a_t^{-i})-\sum_{t=1}^T\log \frac{\rho(a^{-i}_t|s_t)}{P(a^{-i}_t|s_t, o^{-i}_t)} + \sum_{t=1}^T\log P(a^{i}_t|s_t, a^{-i}_t, o^{-i}_t)] \\
    &= \mathbb{E}_{(a^i_{1:T}, a^{-i}_{1:T}, s_{1:T}\sim q)}[\sum_{t=1}^TR^i(s_t, a_t^i, a_t^{-i})-\log \pi(a^i_t|s_t, a_t^{-i})-\log \frac{\rho(a^{-i}_t|s_t)}{P(a^{-i}_t|s_t, o^{-i}_t)} +\log 1]\\
    &= \sum_t\mathbb{E}_{(s_t, a^i_t, a^{-i}_t)\sim q}[R^i(s_t, a_t^i, a_t^{-i})+H(\pi(a^i_t|s_t, a_t^{-i}))-D _ { \mathrm { KL } }(\rho(a^{-i}_t|s_t)||P(a^{-i}_t|s_t, o^{-i}_t))].
\end{align}


\section{Multi-Agent Soft-Q Learning}
\subsection{Soft Q-Function}
\label{appendix:softq}
We define the soft state-action value function $Q^{\pi,\rho}_{soft}(s, a, a^{-i})$ of agent $i$ in a stochastic game as: 
\begin{align}
    &Q^{\pi,\rho}_{soft}(s_t, a^i_t, a^{-i}_t) \nonumber \\
    &= r_t+\mathbb{E}_{(s_{t+l},a^{i}_{t+l}, a^{-i}_{t+l}, \hdots)\sim q}[\sum_{l=1}^{\infty}\gamma^{l}(r_{t+l}+\alpha H(\pi(a^{i}_{t+l}|a_{t+l}^{-i}, s_{t+l}))-D _ { \mathrm { KL } }(\rho(a^{-i}_{t+l}|s_{t+l})||P(a_{t+l}^{-i}|s_{t+l}))] \\
    &= \mathbb{E}_{(s_{t+1},a^{i}_{t+1}, a^{-i}_{t+1})}[ r_t + \gamma (\alpha H(\pi(a^{i}_{t+1}|s_{t+1}, a^{-i}_{t+1}))-D _ { \mathrm { KL } }(\rho(a^{-i}|s_{t+1})||P(a^{-i}|s_{t+1}))+{Q^{\pi,\rho}_{soft}(s_{t+1}, a^{i}_{t+1}, a^{-i}_{t+1})})]\\
    &= \mathbb{E}_{(s_{t+1}, a^{-i}_{t+1})}[ r_t + \gamma (\alpha H(\pi(\cdot|s_{t+1}, a^{-i}_{t+1}))-D _ { \mathrm { KL } }(\rho(a^{-i}|s_{t+1})||P(a^{-i}|s_{t+1}))+\mathbb{E}_{a^{i}_{t+1}\sim \pi}[Q^{\pi,\rho}_{soft}(s_{t+1}, a^{i}_{t+1}, a^{-i}_{t+1})])] \\
    &= \mathbb{E}_{(s_{t+1})}[ r_t + \gamma (\mathbb{E}_{a^{-i}_{t+1}\sim \rho ,a^{i}_{t+1}\sim \pi}[\alpha H(\pi(a^{i}_{t+1}|s_{t+1}, a^{-i}_{t+1}))]-D _ { \mathrm { KL } }(\rho(\cdot|s_{t+1})||P(\cdot|s_{t+1}))] \nonumber \\
    &+\mathbb{E}_{a^{-i}_{t+1}\sim \rho, a^{i}_{t+1}\sim \pi}[Q^{\pi,\rho}_{soft}(s_{t+1}, a^{i}_{t+1}, a^{-i}_{t+1})])],
\end{align}

Then we can easily see that the objective in Eq. \ref{eq:novelobjective} can be rewritten as:
\begin{align}
    \mathcal{J}(\pi, \phi) = \sum_t\mathbb{E}_{(s_t, a^i_t, a^{-i}_t)\sim(p_s,\pi,\rho)}[Q^{\pi,\rho}_{soft}(s_t, a^i_t, a^{-i}_t)+\alpha H(\pi(a^i_t|s_t, a_t^{-i}))-D _ { \mathrm { KL } }(\rho(a^{-i}_t|s_t)||P(a^{-i}_t|s_t))],
\end{align}
by setting $\alpha=1$.

\subsection{Policy Improvement and Opponent Model Improvement}
\label{appendix:policyimprovementandopponentmodelimprovement}
\begin{theorem} (\textit{Policy improvement theorem})
Given a conditional policy $\pi$ and opponent model $\rho$, define a new conditional policy $\Tilde{\pi}$ as 
\begin{equation}
    \Tilde{\pi}(\cdot|s, a^{-i})\propto \exp (\frac{1}{\alpha}Q^{\pi,\rho}_{soft}(s, \cdot,a^{-i})), \forall{s, a^{-i}}.
\end{equation}
Assume that throughout our computation, Q is bounded and $\sum_{a^i}Q(s, a^i,a^{-i})$ is bounded for any $s$ and $a^{-i}$ (for both $\pi$ and $\Tilde{\pi}$). Then $Q^{\Tilde{\pi}, \rho}_{soft}(s, a^i, a^{-i})\geq Q^{\pi,\rho}_{soft}(s, a^i, a^{-i}) \forall{s, a}.$\label{theorem:policyimprove}
\end{theorem}

\begin{theorem} (\textit{Opponent model improvement theorem})
Given a conditional policy $\pi$ and opponent model $\rho$, define a new opponent model $\Tilde{\rho}$ as 
\begin{equation}
    \Tilde{\rho}(\cdot|s)\propto \exp (\sum_{a^i}Q^{\pi,\rho}_{soft}(s, a^i,\cdot)\pi(a^i|\cdot, s)+\alpha H(\pi(s))+\log P(\cdot|s)), \forall{s, a^{i}}.
\end{equation}
Assume that throughout our computation, Q is bounded and $\sum_{a^{-i}}\exp (\sum_{a^i}Q(s, a^i,a^{-i})\pi(a^i|s, a^{-i}))$ is bounded for any $s$ and $a^{i}$ (for both $\rho$ and $\Tilde{\rho}$). Then $Q^{\pi,\Tilde{\rho}}_{soft}(s, a^i, a^{-i})\geq Q^{\pi, \rho}_{soft}(s, a^i, a^{-i}) \forall{s, a}.$\label{theorem:opponentimprove}
\end{theorem}

The proof of Theorem \ref{theorem:policyimprove} and \ref{theorem:opponentimprove} is based on two observations that:
\begin{equation}
    \alpha H(\pi(\cdot|s, a^{-i}))+\mathbb{E}_{a^{i}\sim \pi}[Q^{\pi,\rho}_{soft}(s, a^{i}, a^{-i})] \leq \alpha H(\Tilde{\pi}(\cdot|s, a^{-i}))+\mathbb{E}_{a^{i}\sim \Tilde{\pi}}[Q^{\pi,\rho}_{soft}(s, a^{i}, a^{-i})], \label{equation:ob1}
\end{equation} and
\begin{align}
    &\mathbb{E}_{a^{-i}_{t+1}\sim \rho ,a^{i}_{t+1}\sim \pi}[\alpha H(\pi(a^{i}_{t+1}|s_{t+1}, a^{-i}_{t+1}))]-D _ { \mathrm { KL } }(\rho(\cdot|s_{t+1})||P(\cdot|s_{t+1}))] +\mathbb{E}_{a^{-i}_{t+1}\sim \rho, a^{i}_{t+1}\sim \pi}[Q^{\pi,\rho}_{soft}(s_{t+1}, a^{i}_{t+1}, a^{-i}_{t+1})] \\
    &\leq \mathbb{E}_{a^{-i}_{t+1}\sim \Tilde{\rho} ,a^{i}_{t+1}\sim \pi}[\alpha H(\pi(a^{i}_{t+1}|s_{t+1}, a^{-i}_{t+1}))]-D _ { \mathrm { KL } }(\Tilde{\rho}(a^{-i}_{t+1}|s_{t+1})||P(\cdot|s_{t+1})) +\mathbb{E}_{a^{-i}_{t+1}\sim \Tilde{\rho}, a^{i}_{t+1}\sim \pi}[Q^{\pi,\rho}_{soft}(s_{t+1}, a^{i}_{t+1}, a^{-i}_{t+1})]\label{equation:ob2}.
\end{align}

First, we notice that 
\begin{equation}
    \alpha H(\pi(\cdot|s, a^{-i}))+\mathbb{E}_{a^{i}\sim \pi}[Q^{\pi,\rho}_{soft}(s, a^{i}, a^{-i})] = -\alpha D _ { \mathrm { KL } }(\pi(\cdot|s, a^{-i})||\Tilde{\pi}(\cdot|s, a^{-i})) +\alpha \log  \sum_{a^i}\exp (\frac{1}{\alpha}Q^{\pi,\rho}_{soft}(s, a^i, a^{-i})).
\end{equation}
Therefore, the LHS is only maximized if the KL-Divergence on the RHS is minimized. This KL-Divergence is minimized only when $\pi = \Tilde{\pi}$, which proves the Equation \ref{equation:ob1}.

Similarly, we can have 
\begin{align}
    &\mathbb{E}_{a^{-i}\sim \rho ,a^{i}\sim \pi}[\alpha H(\pi(a^{i}|s, a^{-i}))]-D _ { \mathrm { KL } }(\rho(\cdot|s)||P(\cdot|s))]) +\mathbb{E}_{a^{-i}\sim \rho, a^{i}\sim \pi}[Q^{\pi,\rho}_{soft}(s, a^{i}, a^{-i})] \nonumber \\
    &= -D _ { \mathrm { KL } }(\rho(\cdot|s)||\Tilde{\rho}(\cdot|s))+\log \sum_{a^{-i}}\exp (\sum_{a^i}Q^{\pi,\rho}(s, a^i, a^{-i})\pi(a^i|s, a^{-i})+\alpha H(\pi(\cdot|s, a^{\-i}))+\log P(a^{-i|s})),
\end{align}
which proves the Equation \ref{equation:ob2}.

With the above observations, the proof of Theorem \ref{theorem:policyimprove} and \ref{theorem:opponentimprove} is completed by as follows:
\begin{align}
    &Q^{\pi,\rho}_{soft}(s_t, a^i_t, a^{-i}_t) \nonumber \\
    &= \mathbb{E}_{(s_{t+1},a^{i}_{t+1}, a^{-i}_{t+1})}[ r_t + \gamma (\alpha H(\pi(a^{i}_{t+1}|s_{t+1}, a^{-i}_{t+1}))-D _ { \mathrm { KL } }(\rho(a^{-i}_{t+1}|s_{t+1})||P(a^{-i}_{t+1}|s_{t+1}))+Q^{\pi,\rho}_{soft}(s_{t+1}, a^{i}_{t+1}, a^{-i}_{t+1}))]\\
    &= \mathbb{E}_{(s_{t+1}, a^{-i}_{t+1})}[ r_t + \gamma (\alpha H(\pi(\cdot|s_{t+1}, a^{-i}_{t+1}))-D _ { \mathrm { KL } }(\rho(a^{-i}_{t+1}|s_{t+1})||P(a^{-i}_{t+1}|s_{t+1}))+\mathbb{E}_{a^{i}_{t+1}\sim \pi}[Q^{\pi,\rho}_{soft}(s_{t+1}, a^{i}_{t+1}, a^{-i}_{t+1})])] \\
    &\leq \mathbb{E}_{(s_{t+1}, a^{-i}_{t+1})}[ r_t + \gamma (\alpha H(\Tilde{\pi}(\cdot|s_{t+1}, a^{-i}_{t+1}))-D _ { \mathrm { KL } }(\rho(a^{-i}_{t+1}|s_{t+1})||P(a^{-i}_{t+1}|s_{t+1}))+\mathbb{E}_{a^{i}_{t+1}\sim \Tilde{\pi}}[Q^{\pi,\rho}_{soft}(s_{t+1}, a^{i}_{t+1}, a^{-i}_{t+1})])] \\
    &= \mathbb{E}_{(s_{t+1})}[ r_t + \gamma (\mathbb{E}_{a^{-i}_{t+1}\sim \rho, a^i_{t+1}\sim \pi}[\alpha H(\Tilde{\pi}(a^{i}_{t+1}|s_{t+1}, a^{-i}_{t+1}))]-D _ { \mathrm { KL } }(\rho(\cdot|s_{t+1})||P(\cdot|s_{t+1}))\nonumber\\
    &+\mathbb{E}_{a^{-i}_{t+1}\sim \rho, a^i_{t+1}\sim \pi}[Q^{\pi,\rho}_{soft}(s_{t+1}, a^{i}_{t+1}, a^{-i}_{t+1})])] \\
    &\leq \mathbb{E}_{(s_{t+1})}[ r_t + \gamma (\mathbb{E}_{a^{-i}_{t+1}\sim \Tilde{\rho}, a^i_{t+1}\sim \pi}[\alpha H(\Tilde{\pi}(a^{i}_{t+1}|s_{t+1}, a^{-i}_{t+1}))]-D _ { \mathrm { KL } }(\Tilde{\rho}(\cdot|s_{t+1})||P(\cdot|s_{t+1}))\nonumber\\
    &+\mathbb{E}_{a^{-i}_{t+1}\sim \Tilde{\rho}, a^i_{t+1}\sim \pi}[Q^{\pi,\rho}_{soft}(s_{t+1}, a^{i}_{t+1}, a^{-i}_{t+1})])] \\
    &= \mathbb{E}_{(s_{t+1},a^{i}_{t+1}, a^{-i}_{t+1})\sim\Tilde{q}}[ r_t + \gamma (\alpha H(\Tilde{\pi}(a^{i}_{t+1}|s_{t+1}, a^{-i}_{t+1}))-D _ { \mathrm { KL } }(\Tilde{\rho}(a^{-i}|s_{t+1})||P(a^{-i}|s_{t+1}))+r_{t+1}) \nonumber \\
    &+\gamma ^2 \mathbb{E}_{(s_{t+2}, a^{-i}_{t+2})}[\alpha H(\pi(\cdot|s_{t+2}, a^{-i}_{t+2}))-D _ { \mathrm { KL } }(\rho(a^{-i}_{t+2}|s_{t+2})||P(a^{-i}_{t+2}|s_{t+2}))+\mathbb{E}_{a^i_{t+2}\sim \pi}[Q^{\pi,\rho}_{soft}(s_{t+2}, a^{i}_{t+2}, a^{-i}_{t+2})]]] \\
    &\leq \mathbb{E}_{(s_{t+1},a^{i}_{t+1}, a^{-i}_{t+1})}[ r_t + \gamma (\alpha H(\Tilde{\pi}(a^{i}_{t+1}|s_{t+1}, a^{-i}_{t+1}))-D _ { \mathrm { KL } }(\Tilde{\rho}(a^{-i}|s_{t+1})||P(a^{-i}|s_{t+1}))+r_{t+1}) \nonumber \\
    &+\gamma ^2 \mathbb{E}_{(s_{t+2}, a^{-i}_{t+2})}[\alpha H(\pi(\cdot|s_{t+2}, a^{-i}_{t+2}))-D _ { \mathrm { KL } }(\rho(a^{-i}_{t+2}|s_{t+2})||P(a^{-i}_{t+2}|s_{t+2}))+\mathbb{E}_{a^i_{t+2}\sim \Tilde{\pi}}[Q^{\pi,\rho}_{soft}(s_{t+2}, a^{i}_{t+2}, a^{-i}_{t+2})]]] \\
    &\vdots \nonumber \\
    &\leq r_t+\mathbb{E}_{(s_{t+l},a^{i}_{t+l}, a^{-i}_{t+l}, \hdots)\sim \Tilde{q}}[\sum_{l=1}^{\infty}\gamma^{l}(r_{t+l}+\alpha H(\Tilde{\pi}(a^{i}_{t+l}|a_{t+l}^{-i}, s_{t+l}))-D _ { \mathrm { KL } }(\Tilde{\rho}(a^{-i}_{t+l}|s_{t+l})||P(a_{t+l}^{-i}|s_{t+l}))] \\
    &= Q^{\Tilde{\pi}, \Tilde{\rho}}_{soft}(s_t, a^i_t, a^{-i}_t).
\end{align}

With Theorem \ref{theorem:policyimprove} and \ref{theorem:opponentimprove} and the above inequalities, we can see that, if we start from an arbitrary conditional policy $\pi_0$ and an arbitrary opponent model $\rho_0$ and we iterate between policy improvement as 
\begin{equation}
    \pi_{i+1}(\cdot|s, a^{-i})\propto \exp (\frac{1}{\alpha}Q^{\pi_{t},\rho_{t}}_{soft}(s, \cdot,a^{-i})), 
\end{equation}
and opponent model improvement as
\begin{equation}
        \rho_{t+1}(\cdot|s)\propto \exp (\sum_{a^i}Q^{\pi_{t+1},\rho_{t}}_{soft}(s, a^i,\cdot)\pi_{t+1}(a^i|\cdot, s)+\alpha H(\pi_{t+1}(s))+\log P(\cdot|s)),
\end{equation}
then $Q^{\pi_{t},\rho_{t}}_{soft}(s, a^i,a^{-i})$ can be shown to increase monotonically. Similar to \cite{HaarnojaTAL17rldep}, we can show that with certain regularity conditions satisfied, any non optimal policy and opponent model can be improved this way and Theorem \ref{theorem:optimalpiandrho} is proved.

\subsection{Soft Bellman Equation}
\label{appendix:softbellmanequation}
As we show in Appendix \ref{appendix:policyimprovementandopponentmodelimprovement}, when the training converges, we have: 
\begin{align}
\label{eq:optimalpolicyinappendix}
    \pi^*(a^i|s, a^{-i})=\frac{\frac{1}{\alpha}\exp (Q^*(s, a^i, a^{-i}))}{\sum_{a^i}\exp (\frac{1}{\alpha}Q^*(s, a^i, a^{-i}))},
\end{align}
and
\begin{align}
\label{eq:optimalopponentinappendix}
    \rho^*(a^{-i}|s)&=\frac{\exp (\sum_{a^i}Q^*(s,a^i, a^{-i})\pi^*(a^i|s,a^{-i})+\alpha H(\pi^*(a^i|s, a^{-i}))+\log P(a^{-i}|s))}{\sum_{a^{-i}}\exp (\sum_{a^i}Q^*(s,a^i, a^{-i})\pi^*(a^i|s,a^{-i})+\alpha H(\pi^*(a^i|s, a^{-i}))+\log P(a^{-i}|s))}\nonumber\\
    &=\frac{P(a^{-i}|s)\left(\sum_{a^i}\exp (Q^*_{soft}(s, a^i, a^{-i}))\right)^\alpha}{\exp (V^*(s))},
\end{align}
where the equality in Eq. \ref{eq:optimalopponentinappendix} comes from substituting $\pi^*$ with Eq. \ref{eq:optimalpolicyinappendix} and we define the soft sate value function $V^{\pi,\rho}_{soft}(s)$ of agent $i$ as:
\begin{equation}
    V^{\pi,\rho}_{soft}(s_t) = \log \sum_{a^{-i}_t}P(a^{-i}_t|s_t)\left(\sum_{a^i_t}\exp \left(\frac{1}{\alpha}Q^{\pi,\rho}_{soft}(s_t, a^i_t, a^{-i}_t)\right)\right)^\alpha.
\end{equation}

Then we can show that
\begin{align}
    &Q^{\pi^*,\rho^*}_{soft}(s, a^i, a^{-i}) \nonumber \\
    &=r_t+\gamma\mathbb{E}_{s'\sim p_s}[(\mathbb{E}_{a^{-i}_{t+1}\sim \rho ,a^{i}_{t+1}\sim \pi}[\alpha H(\pi(a^{i}_{t+1}|s_{t+1}, a^{-i}_{t+1}))]-D _ { \mathrm { KL } }(\rho(\cdot|s_{t+1})||P(\cdot|s_{t+1}))] \nonumber \nonumber \\
    &+\mathbb{E}_{a^{-i}_{t+1}\sim \rho, a^{i}_{t+1}\sim \pi}[Q^{\pi,\rho}_{soft}(s_{t+1}, a^{i}_{t+1}, a^{-i}_{t+1})])]\nonumber \\
    &=r_t+\gamma\mathbb{E}_{s'\sim p_s}[V^*(s')].
\end{align}

We define the soft value iteration operator $\mathcal{T}$ as:
\begin{align}
    \mathcal{T}Q(s, a^i, a^{-i}) = R(s, a^i,a^{-i}) + \gamma\mathbb{E}_{s'\sim p_s}\left[\log \sum_{a^{-i\prime}}P(a^{-i\prime}|s^\prime)\left(\sum_{a^{i\prime}}\exp \left(\frac{1}{\alpha}Q(s^\prime, a^{i\prime}, a^{-i\prime})\right)\right)^\alpha\right].
\end{align}

In a symmetric fully cooperative game with only one global optimum,  we can show as done in \cite{wen2018probabilistic}, the operator defined above is a contraction mapping. We define a norm on Q-values $\left\| Q _ { 1 } ^ { i } - Q _ { 2 } ^ { i } \right\| \overset{\Delta}{=}\max _ { s , a ^ { i } , a ^ { - i } } \left| Q _ { 1 } ^ { i } \left( s , a ^ { i } , a ^ { - i } \right) - Q _ { 2 } ^ { i } \left( s , a ^ { i } , a ^ { - i } \right) \right|$. Let $\varepsilon = \left\| Q _ { 1 } ^ { i } - Q _ { 2 } ^ { i } \right\|$, then we have:

\begin{align}
    \log \sum_{a^{-i\prime}}P(a^{-i\prime}|s^\prime)\left(\sum_{a^{i\prime}}\exp \left(\frac{1}{\alpha}Q_1(s^\prime, a^{i\prime}, a^{-i\prime})\right)\right)^\alpha &\leq \log \sum_{a^{-i\prime}}P(a^{-i\prime}|s^\prime)\left(\sum_{a^{i\prime}}\exp \left(\frac{1}{\alpha}Q_2(s^\prime, a^{i\prime}, a^{-i\prime})+\varepsilon\right)\right)^\alpha \nonumber \\
    &= \log \sum_{a^{-i\prime}}P(a^{-i\prime}|s^\prime)\left(\sum_{a^{i\prime}}\exp \left(\frac{1}{\alpha}Q_2(s^\prime, a^{i\prime}, a^{-i\prime})\right)\exp (\varepsilon)\right)^\alpha \nonumber\\
    &= \log \sum_{a^{-i\prime}}P(a^{-i\prime}|s^\prime)\exp (\varepsilon)^\alpha\left(\sum_{a^{i\prime}}\exp \left(\frac{1}{\alpha}Q_2(s^\prime, a^{i\prime}, a^{-i\prime})\right)\right)^\alpha\nonumber\\
    &=\alpha\varepsilon+\log \sum_{a^{-i\prime}}P(a^{-i\prime}|s^\prime)\left(\sum_{a^{i\prime}}\exp \left(\frac{1}{\alpha}Q_2(s^\prime, a^{i\prime}, a^{-i\prime})\right)\right)^\alpha.
\end{align}
Similarly, $\log \sum_{a^{-i\prime}}P(a^{-i\prime}|s^\prime)\left(\sum_{a^{i\prime}}\exp \left(\frac{1}{\alpha}Q_1(s^\prime, a^{i\prime}, a^{-i\prime})\right)\right)^\alpha \geq -\alpha\varepsilon +\log \sum_{a^{-i\prime}}P(a^{-i\prime}|s^\prime)\left(\sum_{a^{i\prime}}\exp \left(\frac{1}{\alpha}Q_2(s^\prime, a^{i\prime}, a^{-i\prime})\right)\right)^\alpha.$ Therefore $\left\| \mathcal{T} Q _ { 1 } ^ { i } - \mathcal{T} Q _ { 2 } ^ { i } \right\| \leq \gamma \varepsilon = \gamma \left\| Q _ { 1 } ^ { i } - Q _ { 2 } ^ { i } \right\|$, where $\alpha=1$.



\end{document}